\documentclass[10,journal]{IEEEtran}
\usepackage{xcolor}
\usepackage{graphicx}
\usepackage[margin=0.75in]{geometry}

\usepackage{amsmath,amssymb,amsthm}
\usepackage{mathtools}
\usepackage{bm}

\usepackage{booktabs}
\usepackage{multirow}
\usepackage{siunitx}

\usepackage{enumitem}
\setlist{nosep}
\setitemize{leftmargin=*}
\usepackage{titlesec}

\usepackage{standalone}

\usepackage{tikz}
\usetikzlibrary{
  3d,
  arrows.meta,
  backgrounds,
  calc,
  positioning,
  shapes.geometric,
  shadows.blur
}
\usepackage{pgfplots}
\pgfplotsset{compat=1.18}
\usepgfplotslibrary{groupplots}

\usepackage[ruled,vlined,linesnumbered]{algorithm2e}
\SetKw{KwBreak}{break}

\definecolor{teal}{RGB}{0,128,128}
\definecolor{slate}{RGB}{112,128,144}
\definecolor{mainblue}{RGB}{25,70,140}
\definecolor{softgray}{RGB}{140,140,140}
\definecolor{aliceblue}{RGB}{240,248,255}

\newcommand{\bd}{\ensuremath{\bm{d}}}
\newcommand{\be}{\ensuremath{\bm{e}}}
\newcommand{\bfm}{\ensuremath{\bm{f}}} 
\newcommand{\bg}{\ensuremath{\bm{g}}}
\newcommand{\bh}{\ensuremath{\bm{h}}}

\newcommand{\bo}{\ensuremath{\bm{o}}}
\newcommand{\bp}{\ensuremath{\bm{p}}}
\newcommand{\bq}{\ensuremath{\bm{q}}}
\newcommand{\br}{\ensuremath{\bm{r}}}
\newcommand{\bs}{\ensuremath{\bm{s}}}

\newcommand{\bu}{\ensuremath{\bm{u}}}

\newcommand{\bx}{\ensuremath{\bm{x}}}
\newcommand{\by}{\ensuremath{\bm{y}}}
\newcommand{\bz}{\ensuremath{\bm{z}}}

\newcommand{\bA}{\ensuremath{\bm{A}}}
\newcommand{\bB}{\ensuremath{\bm{B}}}
\newcommand{\bC}{\ensuremath{\bm{C}}}
\newcommand{\bD}{\ensuremath{\bm{D}}}
\newcommand{\bE}{\ensuremath{\bm{E}}}

\newcommand{\bI}{\ensuremath{\bm{I}}}
\newcommand{\bJ}{\ensuremath{\bm{J}}}
\newcommand{\bK}{\ensuremath{\bm{K}}}
\newcommand{\bL}{\ensuremath{\bm{L}}}

\newcommand{\bN}{\ensuremath{\bm{N}}}

\newcommand{\bR}{\ensuremath{\bm{R}}}
\newcommand{\bS}{\ensuremath{\bm{S}}}

\newcommand{\bV}{\ensuremath{\bm{V}}}
\newcommand{\bW}{\ensuremath{\bm{W}}}
\newcommand{\bX}{\ensuremath{\bm{X}}}

\newtheorem{assumption}{Assumption}
\newtheorem{remark}{Remark}
\newtheorem{theorem}{Theorem}

\usepackage{standalone}
\usepgfplotslibrary{fillbetween,groupplots}

\definecolor{darkblue}{RGB}{0, 0, 139}

\usepackage[colorlinks=true,
  linkcolor=mainblue,
  citecolor=mainblue,
  urlcolor=mainblue]{hyperref}
\newtagform{blue_eq}{\color{mainblue}(}{)}

\usetagform{blue_eq}
\usepackage{lettrine}
\usepackage{cite}
\usepackage{xcolor, colortbl}
\usepackage{graphicx}
\newcommand{\blueRule}{\arrayrulecolor{mainblue}\specialrule{1.2pt}{0.5pt}{0.5pt}\arrayrulecolor{black}}
\newcommand{\thinBlueRule}{\arrayrulecolor{mainblue!50}\specialrule{0.5pt}{0.5pt}{0.5pt}\arrayrulecolor{black}}

\makeatletter
\renewcommand{\section}{\@startsection{section}{1}{\z@}%
  {1.0ex plus 0.6ex minus 0.4ex}%
  {0.4ex plus 0.3ex minus 0.1ex}%
  {\normalfont\large \centering\scshape\color{mainblue}}}

\renewcommand{\subsection}{\@startsection{subsection}{2}{\z@}%
  {0.8ex plus 0.5ex minus 0.3ex}%
  {0.3ex plus 0.2ex minus 0.1ex}%
  {\normalfont\normalsize\itshape\color{mainblue}}}

\renewcommand{\subsubsection}{\@startsection{subsubsection}{3}{\z@}%
  {0.7ex plus 0.4ex minus 0.2ex}%
  {0.2ex plus 0.2ex minus 0.1ex}%
  {\normalfont\normalsize\itshape\color{mainblue}}}
\makeatother

\AtBeginDocument{%
  \setlength{\abovedisplayskip}{3.7pt}%
  \setlength{\belowdisplayskip}{3.7pt}%
  \setlength{\abovedisplayshortskip}{3.5pt}%
  \setlength{\belowdisplayshortskip}{3.5pt}%
  \setlength{\textfloatsep}{8pt plus 2pt minus 2pt}%
  \setlength{\dbltextfloatsep}{8pt plus 2pt minus 2pt}%
  \setlength{\floatsep}{6pt plus 2pt minus 2pt}%
  \setlength{\dblfloatsep}{6pt plus 2pt minus 2pt}%
  \setlength{\intextsep}{6pt plus 2pt minus 2pt}%
  \setlength{\abovecaptionskip}{4pt}%
  \setlength{\belowcaptionskip}{0pt}%
}

\usepackage{makecell}

\usepackage{tabularx}

\usepackage{balance}

\title{\color{mainblue}{\textsc{\textbf{\LARGE Stability and Controllability Margins of Water Network Hydraulics}}}}

\author{Ahmad F. Taha and Mohamad H. Kazma\vspace{-0.8cm}
\thanks{The authors are affiliated with the Civil and Environmental Engineering Department at Vanderbilt University, Nashville, Tennessee. Emails: \url{ahmad.taha@vanderbilt.edu} and \url{mohamad.h.kazma@vanderbilt.edu}.}}
\date{\today}

\begin{document}
\maketitle

\begin{abstract}
Pump and valve actuation in water distribution networks (WDNs) affects how pressures, flows, and tank levels respond to demand changes and other disturbances. This paper presents a stability and controllability analysis of WDN hydraulics. Starting from a general nonlinear differential algebraic equation (DAE) model of WDNs with arbitrary topology and network components (valves and pumps), we investigate three questions. First, we study local well-posedness of the network dynamics and characterize the loss of differentiability introduced by pump and valve switching. Second, we introduce regularization methods that smooth flow and pressure trajectories under changing controls. Third, we quantify how network parametric uncertainty impacts local stability and controllability, while establishing error bounds on the resulting margins. For water operations, stability means that pressures and flows return to steady values after a disturbance, and controllability means that pump and valve actions can move pressures, flows, and tank levels; the margins we compute measure how far a network is from losing either. We demonstrate that the proposed DAE models produce trajectories that closely match EPANET, a widely used WDN simulator, for benchmark networks. The case studies also show that the WDN DAE relates the recovery of pressures and flows to a weighted Laplacian of the network, ranks pipes by their effect on the local dynamics, and reveals that certain demand variation changes the stability and controllability margins without eliminating local stability or the ability of pumps to affect pressures and flows. The developed theoretical foundations enable network analysis, mitigation strategies, and system design.
\end{abstract}


\section{Introduction and Paper Contributions}\label{sec:intro}

\lettrine{W}{hen} a pump changes speed or a valve is actuated, the physics that evolve within a water distribution network (WDN) change, and so must the equations that define its hydraulic state. Beneath the familiar map of pipes is an actuated hydraulic system in which reservoirs, tanks, pumps, valves, demands, and pressure constraints interact through nonlinear head loss and energy balance equations~\cite{todini1988gradient,rossman2020epanet}. This systems interpretation turns common operational decisions into questions about stability and controllability. In the context of water operations, stability refers to whether the network returns to steady pressures and flows after a disturbance and how quickly it does so, while controllability refers to the ability of pump and valve changes to affect pressures, flows, and tank levels, which we refer to as the \textit{reach} of the pumps. From this perspective, we can pose and answer a myriad of questions, such as \textit{How does pump and valve actuation change the hydraulic state within a network? When does an operating point remain stable? Which links are vulnerable to demand that induces flow reversal, and which pressure or tank states can be influenced by the available actuators?} Answering these questions requires developing a model where pump and valve actuation, tank storage, and the hydraulic equations satisfy the mass and energy balance equations while also resulting in an analysis that indicates how far an operating point is from losing stability or controllability. 

\vspace{0.2cm}

Large-scale infrastructure systems such as power grids and transportation networks have well-developed control-theoretic foundations, including stability, optimal control, and decision-making tools. Furthermore, classical hydraulic analysis, state estimation, and control studies are inveterate~\cite{shamir1968analysis,todini1988gradient,walski2003advanced,rossman2020epanet}; however, \textit{a systems analysis of their stability and controllability is less so.} That is, there is an abundance of control algorithms to regulate valves and pumps, but very limited understanding of why WDNs behave as they do. This gap is neither cosmetic nor coincidental. Pumps and valves introduce control inputs, but their setpoints and modes also introduce algebraic constraints and switching events---this naturally forms a nonlinear differential algebraic equation (DAE) hybrid system. The objective of this paper is to write the standard governing equations of WDN hydraulics into a continuous-time nonlinear DAE model whose local regularity, graph structure, stability, and controllability \textit{margins} can be analyzed under demand uncertainty, while retaining the link inertia that a quasi-steady state simulator such as EPANET does not represent. Here, a \textit{margin} measures how far the network is from losing stability or controllability. In this manuscript, rather than invoking nonlinear hybrid systems theory, we use the more scalable and more interpretable DAE linear system theory when possible.  We cover the relevant literature next.

\subsection{Relevant Literature}\label{sec:literature}
\textsc{\textbf{Existing WDN Tools}.} The hydraulic and water quality analysis of WDNs is supported by a plethora of simulation tools. \textit{EPANET} performs extended period hydraulic and water quality simulation, tracking pipe flows, nodal pressures, tank levels, and constituent concentrations over discrete hydraulic time steps~\cite{rossman2020epanet}. \textit{EPyT} mimics EPANET functionality through a Python implementation~\cite{kyriakou2023epyt}. \textit{WNTR} builds a Python workflow compatible with EPANET for resilience studies, disruptive events, repairs, pressure dependent hydraulics, and resilience metrics~\cite{klise2017software}. \textit{WaterModels.jl} provides a Julia/JuMP environment for steady state water flow, optimal water flow, and network design optimization formulations~\cite{tasseff2019watermodels}. These tools are indispensable for reproducible simulation and system control. Their standard modeling abstractions, however, are not intended to answer whether a controlled hydraulic model is a locally regular DAE, where differentiability is lost during device switching, or how stability and controllability margins depend on topology and hydraulic parameters. 

\textsc{\textbf{Steady State and Transient Models}.} A steady state WDN model solves algebraic mass balance and energy balance equations for fixed demands and device settings. This steady state problem is also called the \textit{water flow problem} (WFP). Its convergence properties are studied in~\cite{wang2020new}. Extended period simulation integrates many consecutive WFP solves while updating demand patterns and device statuses between time steps~\cite{todini1988gradient,rossman2020epanet}. At the other end, full hydraulic transient models describe water hammer wave propagation and require compressibility and pipe wall elasticity~\cite{chaudhry2014applied,wylie1993fluid,ghidaoui2005review}. This paper utilizes a model that sits between these levels via a continuous-time rigid water column DAE that retains link inertance and tank storage, but it does \textit{not} resolve acoustic pressure waves. This intermediate regime is deliberate because it is the level at which pump and valve actions, tank dynamics, and hydraulics can be coupled in one tractable model. The rigid water column formulation in~\cite{nault2016improved} targets slow transients and controlled operations, including link inertance and tank storage but not wave propagation. The present DAE extends that level by adding a unified control input description, valve residuals, and a graph-theoretic linearization amenable to stability and controllability analysis.

Recent work revisits real-time dynamic hydraulic modeling, and rigid water column formulations address slow transients and SCADA integration~\cite{nault2016improved,deuerlein2015parameterization}. Pressure-driven network analysis provides complementary algebraic solvability tools~\cite{piller2007unified} and \textit{Port-Hamiltonian} models support passivity based control~\cite{perryman2022port,torres2019port}. Contemporary literature presents spectral techniques that support WDN partitioning and resilience analysis~\cite{dinardo2018applications,yazdani2011complex}. These works motivate the present modeling level, but they do not provide a general DAE formulation that simultaneously tracks link inertance, tank storage, junction and reservoir algebraic constraints, pump curves, valve mode residuals, and the index and regularity conditions needed for subsequent stability and control analysis. 

\textsc{\textbf{Control in Water Networks}.} A substantial WDN control literature optimizes pumps, valves, pressure zones, or operating costs. Representative examples include predictive control for the urban water cycle~\cite{ocampo2013predictive}, geometric programming control for tree and general drinking water networks~\cite{sela2015control,wang2020receding}, economic MPC~\cite{wang2017nonlinear}, geometric programming control with DAE constraints~\cite{wang2019gpdae}, and optimal operations via advanced optimization~\cite{singh2020optimal}. Some closely related recent work also covers WDN control-theoretic studies on stochastic predictive control, distributionally robust pump control, and pressure regulation with stability guarantees~\cite{sampathirao2018gpu,guo2023optimal,galuppini2024realtime}. Within the water resources literature, pump and valve scheduling has also been posed through mixed integer linear programming~\cite{Thomas_2024}. These studies are all relevant, but their primary objective is optimization or control synthesis, not a systems level analysis.

\subsection{Research Gaps, Novelty, and Contributions}\label{sec:contributions}
\textsc{\textbf{Paper Contributions and Organization}.} To the best of our knowledge, no prior work introduces a unified local DAE analysis of WDN hydraulics that connects fixed mode regularity, loss of differentiability induced by switching, graph Laplacian structure, and controllability margin sensitivity. Individual subsets exist, including graph spectral properties of WDNs~\cite{dinardo2018applications,yazdani2011complex}, structural observability~\cite{vangervert2025structural}, and local sensitivity of steady state solutions~\cite{piller2016sensitivity}. That being said, none of the aforementioned results provide a DAE framework for the stability and controllability analysis of WDNs; this paper aims to fill this gap. In particular, the controllability of WDN hydraulics has not been analyzed. Such an analysis offers more insight than solving an optimization problem, as it builds on the physics of the network and thus explains how its layout and hydraulic parameters shape its response to uncertainty and disturbances. From a water operations perspective, the proposed framework can \textit{(i)} identify the pipes whose roughness most affects how quickly pressures recover after a disturbance, thereby indicating where roughness uncertainty matters most and where pressure monitoring is most informative; and \textit{(ii)} quantify how the stability and controllability margins change as demand varies, which indicates whether pump schedules designed at low demand remain effective at peak demand. The paper contributions are as follows.
\begin{enumerate}
\item \textit{Full nonlinear DAE model and local regularity and switching analysis.} We derive a continuous-time DAE that retains pipe, pump, and valve flows, junction, tank, and reservoir heads, tank storage, pump curves, valve modes, and nodal demands in one unreduced state-space form. We then prove that, on fixed mode intervals, the WDN DAE has differentiation index two; 
it has a unique local hydraulic trajectory. Index two means that one differentiation of the nodal mass balance determines the junction heads through a weighted Laplacian of the network. As such, the junction heads depend on the rate of change of the demands. This is not captured by quasi-steady state simulations. We also introduce a smoothed DAE such that each switch is replaced by a continuous transition of the pump speed or valve setting. 
\item \textit{Graph-theoretic local model.} We show that, around an operating point, the smoothed DAE is described by a weighted graph of the network (a weighted Laplacian), where each link is weighted by its local hydraulic conductance. This explicitly connects the network layout with its component properties and thus informs how any uncertain parameters affect the hydraulic conditions within a WDN. We further show that the incremental hydraulic energy of the network decays at a rate set by the link resistances, that is, the operating point is locally asymptotically stable when each junction and tank is linked to a reservoir through a path of open pipes, pumps, or valves. We show that the link resistances and inertances set how quickly pressures and flows recover after a sudden change, such as a pump trip, that quasi-steady state simulations do not resolve.
\item \textit{Stability, controllability, and perturbation margins.} We quantify how hydraulic parameter perturbations affect the local linearization, the stability margin, and finite-horizon controllability margins. We show the relevance of the resulting parameter ranking, \textit{(i)} by identifying the pipes that most affect the local dynamics; and \textit{(ii)} by showing how the controllability quantities inform pump allocation decisions.
\item \textit{Thorough case studies and benchmarking.} We present detailed case studies for various benchmark water networks that demonstrate the operational interpretation of the proposed results while corroborating the theoretical methods.
\end{enumerate}

\textsc{\textbf{Broader Impacts}.} The proposed framework results in local stability and controllability margins of a water network. This characterization is relevant when pipe roughness and demands are uncertain~\cite{deuerlein2015parameterization,berardi2021calibration,diao2019demand} and when pump and valve schedules change over the operation of the network~\cite{singh2020optimal,guo2023optimal,Thomas_2024}. The margins and the pipe ranking obtained herein thereby provide water system operators with the pipes that most affect pressure recovery, the pump and valve actions that retain their influence under demand variation, and the operating conditions under which the network returns to steady pressures and flows after a disturbance~\cite{galuppini2024realtime}. Such a perspective is essential for operating water networks under uncertainty, as it provides reliability measures under changing operating conditions~\cite{klise2017software,sampathirao2018gpu} and identifies the operating conditions that merit monitoring. This perspective also extends to actuator selection~\cite{dai2016complementarity,hien2021complementarity}, where the same quantities can quantify how the allocation of pumps and valves affects the ability to regulate pressures and tank levels, maintain safe operating limits, and recover from disturbances. Anomaly detection can also be studied. A pump trip, a pipe burst, or a manipulated pump or valve setting can produce a pressure transient. The duration of this transient is set by the link time constants, while the slower recovery of tank levels is set by the stability margin. As such, a model that resolves these transients enables \textit{(i)} the detection of bursts and of attacks on the sensing of water networks~\cite{Pasqualetti2013}; and \textit{(ii)} the identification of the nodes that are most informative for placing the sensors that enable detecting such events~\cite{Huang2021}.

\textsc{\textbf{Paper Organization}.} Section~\ref{sec:model_methods} presents the model and methods. It derives the full hydraulic DAE (Section~\ref{sec:methods}), studies local DAE regularity and smooth control regularization (Section~\ref{sec:local_theory_overall}), and develops the systems-theoretic and graph-theoretic stability and controllability results (Section~\ref{sec:graph_theory}). Section~\ref{sec:networks_settings} describes the benchmark networks and the numerical settings. Section~\ref{sec:case_studies} presents and discusses the results. Section~\ref{sec:future} concludes the paper. The notation of the paper is given in Appendix~\ref{app:notation}, and the proofs of the theorems are given in the appendices.

\section{Model and Methods}\label{sec:model_methods}
\subsection{DAE Hydraulic Model and State-Space Representation}\label{sec:methods}
This section formulates the continuous-time hydraulic model used throughout the paper. The governing equations are standard in drinking water hydraulics. The main objective here is to rewrite these equations into an \textit{unreduced} nonlinear DAE that preserves the network structure. This facilitates the subsequent stability and control analysis.  The algebraic hydraulic equations used by classical texts are well established~\cite{bhave2006analysis,walski2003advanced,walski2001water,wang2020new}. The model used here is \emph{not} a purely steady state model; it is a continuous-time rigid water column transient DAE where link flows have inertia, tank heads evolve through storage, and junction and reservoir constraints are enforced algebraically. This modeling level retains the slow hydraulic dynamics needed for control analysis while neglecting compressibility and wave propagation. This is standard for rigid water column treatments in the literature~\cite{coulbeck1978modelling,coulbeck1980dynamic,chaudhry2014applied,masuda2019dynamical}. 
\subsubsection{Network Elements and Set Notation}
Let the set of nodes $\mathcal{V}$ be partitioned as $\mathcal{V}=\mathcal{J}\cup\mathcal{A}\cup\mathcal{R}$, where $\mathcal{J}$, $\mathcal{A}$, and $\mathcal{R}$ denote junctions, tanks, and reservoirs.  Let $n_E$ be the number of active links, and let $n_J$, $n_A$, and $n_R$ be the numbers of junctions, tanks, and reservoirs. The superscript $A$ is used for tank related quantities to avoid confusion with the transpose symbol $(\cdot)^\top$. Let the link set be partitioned as $\mathcal{E}=\mathcal{P}\cup\mathcal{M}\cup\mathcal{W}$, where $\mathcal{P}$, $\mathcal{M}$, and $\mathcal{W}$ denote pipes, pumps, and valves. The full head and flow vectors are $\bp(t)=\{\bp^J(t),\bp^A(t),\bp^R(t)\}$ and $\bq(t)=\{\bq^P(t),\bq^M(t),\bq^W(t)\}$, where $\{\cdot\}$ stacks the listed vectors into one column vector.

\begin{figure*}[t]
	\centering
	\includegraphics[width=\textwidth]{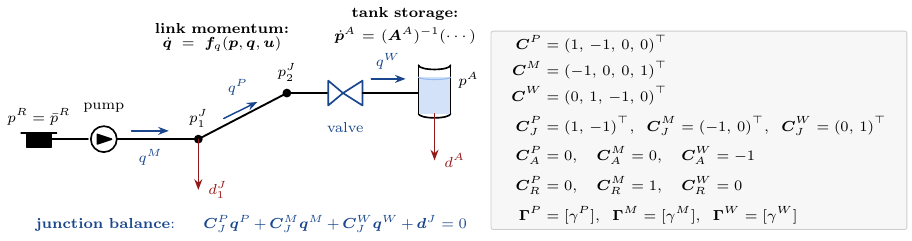}
	\caption{Simple network describing~\eqref{eq:full_hydraulic_dae}. The red arrows denote positive demand, and the blue arrows denote the chosen positive flow directions.}
	\label{fig:full_dae_schematic}
\end{figure*}

For each link type $\alpha\in\{P,M,W\}$, the incidence matrix $\bC^\alpha\in\mathbb{R}^{|\mathcal{V}|\times |\mathcal{E}^\alpha|}$ is defined from a chosen positive direction for each link. If link $e=(i,j)$ is directed from node $i$ to node $j$, then $[\bC^\alpha]_{ke}=+1$ when $k=i$, $[\bC^\alpha]_{ke}=-1$ when $k=j$, and $[\bC^\alpha]_{ke}=0$ otherwise. With this convention, $(\bC^\alpha)^\top\bp$ is the signed head drop from the start node to the end node of a positively oriented link. After ordering node rows as junctions, tanks, and reservoirs, the same matrix is partitioned as $
\bC^\alpha=
\begin{bsmallmatrix}
\bC_J^\alpha\\
\bC_A^\alpha\\
\bC_R^\alpha
\end{bsmallmatrix},
\bC_J^\alpha\in\mathbb{R}^{|\mathcal{J}|\times|\mathcal{E}^\alpha|},$
with analogous dimensions for $\bC_A^\alpha$ and $\bC_R^\alpha$. Thus these blocks are matrices obtained by selecting the appropriate node rows of $\bC^\alpha$, not new vectors. The hydraulic DAE model can be written as
\begin{subequations}\label{eq:full_hydraulic_dae}
\begin{align}
\dot{\bq}^{P}(t) &= \boldsymbol{\Gamma}^{P}\Big((\bC^{P})^\top\bp(t)-\boldsymbol{\varphi}^{P}(\bq^{P}(t))\Big), \label{eq:full_hydraulic_dae_pipes}\\
\dot{\bq}^{M}(t) &= \boldsymbol{\Gamma}^{M}\Big((\bC^{M})^\top\bp(t)+\boldsymbol{\psi}^{M}(\bq^{M}(t),\bs(t))\Big), \label{eq:full_hydraulic_dae_pumps}\\
\dot{\bq}^{W}(t) &= \boldsymbol{\Gamma}^{W}\Big((\bC^{W})^\top\bp(t)-\boldsymbol{\varphi}^{W}(\bq^{W}(t),\bo(t))\Big), \label{eq:full_hydraulic_dae_valves}\\
\dot{\bp}^{A}(t) &= (\bA^{A})^{-1}\Big(-\bC_{A}^{P}\bq^{P}(t)-\bC_{A}^{M}\bq^{M}(t)\nonumber\\ &\qquad-\bC_{A}^{W}\bq^{W}(t)-\bd^{A}(t)\Big), \label{eq:full_hydraulic_dae_tanks}\\
\mathbf{0} &= \bC_{J}^{P}\bq^{P}(t)+\bC_{J}^{M}\bq^{M}(t) +\bC_{J}^{W}\bq^{W}(t)+\bd^{J}(t), \label{eq:full_hydraulic_dae_junctions}\\
\mathbf{0} &= \bp^{R}(t)-\bar{\bp}^{R}. \label{eq:full_hydraulic_dae_reservoirs}
\end{align}
\end{subequations}
Equations~\eqref{eq:full_hydraulic_dae_pipes}-\eqref{eq:full_hydraulic_dae_valves} are link momentum balances, \eqref{eq:full_hydraulic_dae_tanks} is the tank storage equation where $\bA^{A}$ is the diagonal matrix of tank cross sectional areas, \eqref{eq:full_hydraulic_dae_junctions} is the junction mass balance equation, and \eqref{eq:full_hydraulic_dae_reservoirs} enforces the fixed head reservoir condition $p_r^R(t)=\bar p_r^R$ for every reservoir node $r\in\mathcal{R}$. With the incidence convention above, $\bC^\alpha\bq^\alpha$ is outflow minus inflow at each node, thus a positive component of $\bd$ is a demand, that is, water leaving the network. The diagonal entries of $(\boldsymbol{\Gamma}^\alpha)^{-1}$ are link inertances, and $\boldsymbol{\Gamma}^\alpha$ scales the rate at which each link flow responds to a head imbalance. Therefore a link with inertance $I_e$ contributes the scalar multiplier $\gamma_e=1/I_e$ to the corresponding diagonal entry, so larger inertance means slower flow acceleration under the same head imbalance. For a rigid water column, the pipe inertance is $I_e = \ell_e / (g A_e)$, where $\gamma_e = g A_e / \ell_e$ is based on the pipe length and diameter. Pumps and valves receive the median pipe inertance as a default and scaled by user supplied factors~\cite{nault2016improved}. The vector valued nonlinearities in~\eqref{eq:full_hydraulic_dae} are
\begin{subequations}\label{eq:component_nonlinearities}
\begin{align}
\boldsymbol{\varphi}^{P}(\bq^{P}) &= \br^{P}(\bq^{P})\odot \bq^{P}\odot |\bq^{P}|, \label{eq:component_nonlinearities_pipe}\\
\boldsymbol{\varphi}^{W}(\bq^{W},\bo) &= \br^{W}(\bo)\odot \bq^{W}\odot |\bq^{W}|, \label{eq:component_nonlinearities_valve}\\
\boldsymbol{\psi}^{M}(\bq^{M},\bs) &= \{s_m^2\big(h_m^0-r_m^{M}(q_m^M s_m^{-1})^{\nu_m}\big)\}_{m=1}^{n_M}. \label{eq:component_nonlinearities_pump}
\end{align}
\end{subequations}
Equation~\eqref{eq:component_nonlinearities_pipe} is a nonlinear headloss model written in vector form, where $\odot$ denotes the componentwise product and $|\cdot|$ the componentwise absolute value, \eqref{eq:component_nonlinearities_valve} plays the same role for valve losses, and \eqref{eq:component_nonlinearities_pump} is the standard pump curve scaled by the pump speed. In~\eqref{eq:component_nonlinearities_pump}, $h_m^0$ is the shutoff head coefficient, $r_m^M$ is the pump curve resistance coefficient, $\nu_m$ is the pump curve exponent, and $s_m$ is the pump speed. For a Reynolds dependent Darcy-Weisbach description, the pipe resistance entries satisfy
\begin{equation}\label{eq:pipe_resistance_law}
r_e^{P}(q_e^{P})=\frac{8\,\ell_e\,f_e(\mathrm{Re}_e(q_e^{P}))}{g\pi^2 d_e^5},
\qquad
\mathrm{Re}_e(q_e^{P})=\frac{4|q_e^{P}|}{\pi \nu_w d_e},
\end{equation}
where $\ell_e$ and $d_e$ are the pipe length and diameter, $g$ is the gravitational acceleration, $\nu_w$ is the kinematic viscosity of water, $\mathrm{Re}_e$ is the Reynolds number, and $f_e(\cdot)$ is the friction factor~\cite{chaudhry2014applied,rossman2020epanet}. Other standard headloss relations used in WDNs (such as Hazen-Williams) can be similarly embedded in $\br^P(\cdot)$. Fig.~\ref{fig:full_dae_schematic} illustrates how the link, junction, tank, and reservoir components of~\eqref{eq:full_hydraulic_dae} couple through the network incidence and the chosen positive flow directions.

Standard valve classes split into two groups. General purpose valves (GPVs), throttle control valves (TCVs), and pressure breaker valves (PBVs) behave as dissipative links and are represented through the valve headloss equation~\eqref{eq:full_hydraulic_dae_valves}. Flow control valves (FCVs), pressure reducing valves (PRVs), and pressure sustaining valves (PSVs) additionally regulate flow or pressure when active and require mode dependent residuals. Related optimization formulations, which represent the operating modes of a valve through complementarity constraints are introduced for PRV localization and pressure management in~\cite{dai2016complementarity,hien2021complementarity}. All of these valves can be considered within the framework in this paper. Valve modeling details are omitted here for brevity.

\subsubsection{State-Space Modeling}
For compactness, we consider the state as $\bx=\{\bz,\by\}$ with $\bz=\{\bq,\bp^A\}$ (as the \textit{differential state}) and $\by=\{\bp^J,\bp^R\}$ (as the \textit{algebraic state}). This partition assumes zero junction storage, which is standard for WDNs. Here $\bx\in\mathbb{R}^{n_x}$ with $n_x=n_q+n_p$, where $n_q$ is the number of link flows and $n_p$ is the number of nodal heads. The differential and algebraic states satisfy $\bz\in\mathbb{R}^{n_z}$ and $\by\in\mathbb{R}^{n_y}$, with $n_z=n_q+n_A$ and $n_y=n_J+n_R$. The control input $\bu=\{\bs,\bo\}$ collects the pump speeds $\bs$ and the valve settings $\bo$, and the demand vector $\bd=\{\bd^J,\bd^A\}\in\mathbb{R}^{n_d}$ collects the water demands at junctions and tanks.  Then \eqref{eq:full_hydraulic_dae_pipes}-\eqref{eq:full_hydraulic_dae_tanks} assemble into the differential part, while \eqref{eq:full_hydraulic_dae_junctions}-\eqref{eq:full_hydraulic_dae_reservoirs} assemble into the algebraic part as follows
\begin{subequations}\label{eq:compact_dae_form}
\begin{align}
\dot{\bz}(t) &= \bfm(\bx(t),\bu(t),\bd(t)), \label{eq:compact_dae_form_differential}\\
\mathbf{0} &= \bh(\bx(t),\bu(t),\bd(t)). \label{eq:compact_dae_form_algebraic}
\end{align}
\end{subequations}
In this form, $\bfm$ collects the \textit{transient} momentum and tank storage equations (transient here refers to the slow dynamics of a rigid water column, not to the water hammer models mentioned in Section~\ref{sec:literature}; this paper does not delve into the latter). Furthermore, $\bh$ collects the junction and reservoir constraints together with any active regulating valve residuals. A construction of these components for a three node network, including the topology dependent incidence matrices, is given in Appendix~\ref{app:threenodes}. 

The DAE in~\eqref{eq:compact_dae_form} is a slow transient model because it evolves continuously in time and retains link inertance and tank storage~\cite{nault2016improved}. It is nevertheless a rigid water column approximation rather than a full water hammer model due to the fact that pressure waves and fluid compressibility are not resolved~\cite{wylie1993fluid,chaudhry2014applied,ghidaoui2005review}. WDN toolboxes (like EPANET) solve the full nonlinear algebraic system at each hydraulic time step, but they do not explicitly treat link inertance---tank levels are updated between steps, while pipe flows are assumed to reach instantaneous equilibrium~\cite{todini1988gradient,rossman2020epanet}. The DAE model, by contrast, integrates pipe flow dynamics explicitly, so the two models agree closely when inertial transients decay quickly relative to the hydraulic time step.  If the inertial terms $\dot{\bq}^{P}$, $\dot{\bq}^{M}$, and $\dot{\bq}^{W}$ are neglected in~\eqref{eq:full_hydraulic_dae}, the familiar quasi-steady algebraic structure of EPANET is recovered. The DAE should therefore agree closely with EPANET when link inertial transients decay quickly relative to the imposed demand and control variations. The numerical case studies test exactly that.

\subsection{Well-Posedness and Smoothing of Pump and Valve Switching}\label{sec:local_theory_overall}

Water system operation is naturally piecewise constant in time. Pump speeds, valve statuses, and setpoints are typically held fixed over each hydraulic time step and updated only at scheduled times or when a logic condition is triggered, as in extended period simulation~\cite{rossman2020epanet,walski2003advanced} through control rules. This section investigates theoretical properties of the DAE model with and without switching---then derives a smoothed DAE model.

 \subsubsection{Local Well-Posedness Under Fixed Controls}\label{sec:local_theory}
 
In this section, the first question is practical rather than abstract: \textit{ If the controls are frozen over one such interval, does the continuous-time DAE produce a unique local hydraulic trajectory?} Hydraulic simulations routinely compute such trajectories, but existence, uniqueness, and index properties determine when sensitivities, linearizations, and reductions to ordinary differential equations are legitimate. To the best of our knowledge, this question has \textit{not} been stated explicitly for controlled DAE WDN models.
 
 The second question is \textit{how the DAE can be simulated smoothly when pump or valve statuses and setpoints change}. Assumption~\ref{ass:regular_fixed_mode} formalizes the fixed control regime, Theorem~\ref{thm:local_index_one} answers the first question, and Section~\ref{sec:smoothed_switching} answers the second.

\begin{assumption}[Fixed control hydraulic regime]\label{ass:regular_fixed_mode}
Fix a time interval $\mathcal{I}=[t_0,t_1]$ over which all pump speeds, valve modes, and valve setpoints are held constant, and assume the demand $\bd(t)$ is twice continuously differentiable on $\mathcal{I}$. Assume \textit{(i)} the nonlinear pipe, pump, and valve link equations in~\eqref{eq:component_nonlinearities} are continuously differentiable on an open neighborhood of the trajectory of interest, with all active pump speeds satisfying $s_m\ge \underline{s}>0$, \textit{(ii)} all hydraulic parameters and tank areas remain finite and strictly positive, and \textit{(iii)} every junction is connected to a tank or a reservoir through open links. Equivalently, the inertance weighted Laplacian $\bL_J=\bC_J\boldsymbol{\Gamma}\bC_J^\top$ is nonsingular, where $\bC_J=[\bC_J^P\ \bC_J^M\ \bC_J^W]$ and $\boldsymbol{\Gamma}=\operatorname{diag}(\boldsymbol{\Gamma}^P,\boldsymbol{\Gamma}^M,\boldsymbol{\Gamma}^W)$ are restricted to the open links.
\end{assumption}
Assumption~\ref{ass:regular_fixed_mode} is standard and practical in water systems. The continuously differentiable link equations exclude sharp transitions such as the transition from laminar to turbulent flow in the Darcy-Weisbach friction factor. For typical WDN operating points flows are turbulent. The edge case of very small nighttime flows in dead end branches is not addressed here.  Fixed control settings over one hydraulic time step are exactly how most supervisory simulations are executed, while positive tank areas, finite link parameters, and smooth local headloss equations are basic physical conditions~\cite{rossman2020epanet,walski2003advanced,chaudhry2014applied}. The connectivity condition in \textit{(iii)} is a property of the network layout; it does not depend on the operating point or on the head loss equations. It fails when a closed pump or valve isolates a group of junctions from every tank and reservoir. The heads of those junctions are then not determined by the model.

\begin{remark}[Connection with static hydraulic solvability]\label{rem:static_solvability}
The nonsingularity of $\bL_J$ is the DAE analogue of the condition used when the gradient based hydraulic solver of EPANET computes a static solution of the WFP~\cite{todini1988gradient,rossman2020epanet,wang2020new}. At each hydraulic time step, that solver couples the energy equations of the links, where the heads appear, with the mass balance equations at the junctions (flow balance). Solving the DAE~\eqref{eq:full_hydraulic_dae} results in an additional inertia term per link that vanishes when the inertia of the water in the link is negligible. That is, when a junction is not connected to a tank or reservoir, its head is not determined in either model. Thus the steady hydraulic solve is ill-posed before any stability or control analysis is attempted.
\end{remark}

Define a \emph{hydraulic operating point} as the combined vector $\{\bar{\bx},\bar{\bu},\bar{\bd}\}$ satisfying both $\bfm(\bar{\bx},\bar{\bu},\bar{\bd})=\mathbf{0}$ and $\bh(\bar{\bx},\bar{\bu},\bar{\bd})=\mathbf{0}$. The next result discusses uniqueness of hydraulic trajectories. 

\begin{theorem}[Unique hydraulic evolution under fixed controls]\label{thm:local_index_one}
Under Assumption~\ref{ass:regular_fixed_mode}, the hydraulic DAE~\eqref{eq:compact_dae_form} has differentiation index two~\cite[Sec.~4.2]{kunkel2006dae}, and the algebraic variables $\by$ are uniquely determined by the differential variables $\bz$, the inputs $\bd$ and $\bu$. The rate of change of the junction demand $\dot{\bd}^J$ through the reservoir condition~\eqref{eq:full_hydraulic_dae_reservoirs} and
\begin{equation}\label{eq:hidden_constraint}
\bL_J\,\bp^J=-\bC_J\boldsymbol{\Gamma}\big(\bC_A^\top\bp^A+\bC_R^\top\bar{\bp}^R-\boldsymbol{\eta}(\bq,\bu)\big)-\dot{\bd}^J,
\end{equation}
that is obtained by differentiating the junction balance~\eqref{eq:full_hydraulic_dae_junctions} once. Here $\bC_A$, $\bC_R$, and $\boldsymbol{\eta}(\bq,\bu)=\{\boldsymbol{\varphi}^P(\bq^P),-\boldsymbol{\psi}^M(\bq^M,\bs),\boldsymbol{\varphi}^W(\bq^W,\bo)\}$ stack the incidence blocks and the link equations of~\eqref{eq:component_nonlinearities} over the open links. Consequently, every consistent initial condition $\bx(t_0)$ (an initial state that satisfies~\eqref{eq:full_hydraulic_dae_junctions}, \eqref{eq:full_hydraulic_dae_reservoirs}, and~\eqref{eq:hidden_constraint}) generates a unique local DAE solution $\bx(t)$ for $t\in[t_0,t_0+\delta]\subseteq\mathcal{I}$ for some $\delta>0$.
\end{theorem}
The proof is given in Appendix~\ref{app:proof_local_index_one}. The unique DAE solution satisfies the junction balance for all $t$ and can be continued on $\mathcal{I}$ as long as it remains in the neighborhood where Assumption~\ref{ass:regular_fixed_mode} holds. Note that, with the reservoir heads given by~\eqref{eq:full_hydraulic_dae_reservoirs}, the DAE is in Hessenberg form of index two in the junction heads~\cite[Eq.~(4.36)]{kunkel2006dae}. The key implication is that the algebraic part of the model is not independent of the dynamic state variables. On the same neighborhood where Assumption~\ref{ass:regular_fixed_mode} holds, the reservoir condition~\eqref{eq:full_hydraulic_dae_reservoirs} and the hidden constraint~\eqref{eq:hidden_constraint} define a local map $\by=\boldsymbol{\pi}_y(\bz,\bu,\bd,\dot{\bd}^J)$, such that the differential state satisfies
\begin{equation}\label{eq:local_reduced_ode}
\begin{split}\dot{\bz}(t)&=\mathbf{F}\big(\bz(t),\bu(t),\bd(t),\dot{\bd}^J(t)\big)\\&=\bfm\big(\{\bz,\boldsymbol{\pi}_y(\bz,\bu,\bd,\dot{\bd}^J)\},\bu,\bd\big).\end{split}
\end{equation}
The vector field $\mathbf{F}$ is defined on the set of flows that satisfy the junction balance~\eqref{eq:full_hydraulic_dae_junctions}. Given that~\eqref{eq:hidden_constraint} enforces $\bC_J\dot{\bq}=-\dot{\bd}^J$, every solution that starts on that set remains on it. When the controls or demands are prescribed time functions on a fixed mode interval, we use the shorthand $\mathbf{F}(\bz,t)=\mathbf{F}(\bz,\bu,\bd(t),\dot{\bd}^J(t))$.
That said, with fixed pump and valve settings, only the link flows and tank levels evolve in time. Furthermore, a pump or valve status update changes the residual equations, however it does not by itself impose an impulsive reset on pipe flows or tank heads. Thus $\bz(t)$ remains continuous across the event unless a different hybrid model adds a reset map. The exception is a status change that closes a link carrying flow, since that flow must drop to zero at once, which the smoothing below avoids. For example, when a PRV switches from active to open, the downstream head changes from the regulated setpoint to the value determined by the pipe network since the valve no longer regulates the downstream head.

\subsubsection{Smooth Control Regularization for Multistep Simulation}\label{sec:smoothed_switching}

Exact pump and valve switching is not directly usable by a continuous-time integrator. An instantaneous pump or valve switch makes the hydraulic trajectory nondifferentiable at the switching instant, as shown in Appendix~\ref{app:proof_switch_nondiff}. If no reset is applied, pipe flows and tank heads remain continuous. For extended period simulation invoked by changing controls, a standard DAE integrator therefore benefits from a \textit{short regularization} of the control law.  

The regularization idea is simple. One keeps the same scheduled device values, but replaces each abrupt update by a smooth transition over a narrow time window. Let $0=t_0<t_1<\cdots<t_K$ be the control update times, and let $\bar{\bu}_k$ denote the control vector prescribed on $[t_k,t_{k+1})$. Choose a smoothing width $\tau_s$ such that $0<\tau_s<t_{k+1}-t_k$ for every $k$.  On the switching window $[t_k,t_k+\tau_s]$, define
\begin{subequations}\label{eq:smooth_control_interpolation} 
\begin{align}
\bu_{\tau_s}(t)
&=
\big(1-\chi(\alpha_k(t))\big)\bar{\bu}_k
+\chi(\alpha_k(t))\bar{\bu}_{k+1}, \label{eq:smooth_control_interpolation_main}\\
\alpha_k(t)
&=
\frac{t-t_k}{\tau_s}, \qquad
\chi(\alpha)=10\alpha^3-15\alpha^4+6\alpha^5. \label{eq:quintic_smoothstep}
\end{align}
\end{subequations}
\begin{figure*}[t]
	\centering
	\includegraphics[width=\textwidth]{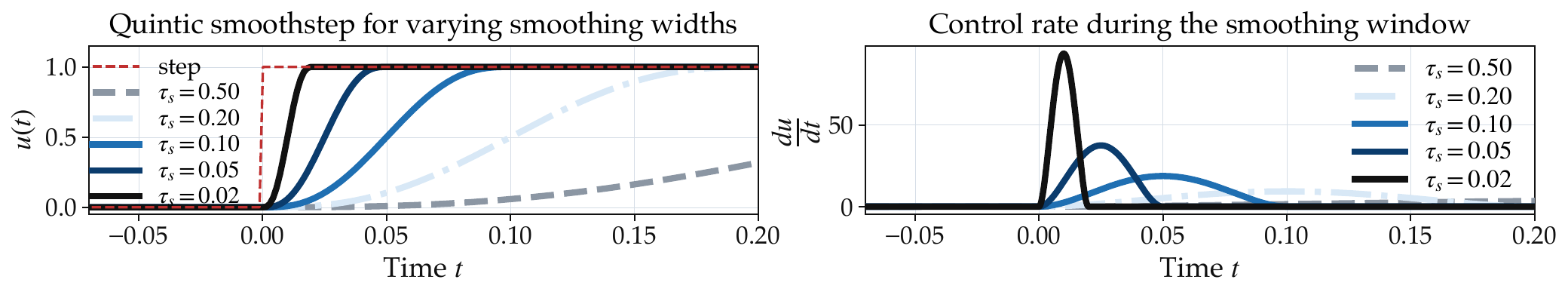}
	\caption{Effect of the smoothing width $\tau_s$ on the quintic transition. The left figure shows control trajectories for several $\tau_s$ values and the right figure shows the control rate $du/dt$.}
	\label{fig:smooth_control_regularization}
\end{figure*}
Outside these windows, $\bu_{\tau_s}(t)$ is held equal to the scheduled value. This construction follows the regularization idea used for piecewise smooth dynamical systems, where a jump is replaced by a continuous transition function over a small neighborhood of the switching surface~\cite{guglielmi2022piecesmooth}. Other strategies for handling valve mode transitions include event driven switching between rigid and elastic formulations~\cite{nault2016improved}, extended period simulation with transient models~\cite{filion2002transient}, and quadratic headloss approximations that are everywhere smooth~\cite{pecci2017quadratic}. Beyond hydraulic simulation, the numerical treatment of differential equations with a discontinuous right hand side is well-studied. Time stepping and event driven schemes integrate the nonsmooth system directly, either by locating each switching instant and restarting the integrator or by writing the switch as a complementarity problem at every step~\cite{acary2008nonsmooth,dieci2012survey}. The approach considered herein is instead a regularization, such that the discontinuous control schedule is replaced by a continuously differentiable one. We adopt it because the linearization, stability, and controllability analyses of Section~\ref{sec:graph_theory} require a differentiable residual while the pump and valve modes are fixed, which an event driven integration does not provide at the switching instant itself. The quintic smoothing used here is simple and differentiable but introduces a tuning parameter $\tau_s$. The specific polynomial considered is the scalar quintic Hermite interpolant with zero first and second endpoint derivatives, a standard choice when one wants a twice differentiable transition between successive setpoints~\cite{lind2020quintic}. It is obtained by solving the six endpoint conditions
\[
\begin{gathered}\chi(0)=0,\quad \chi(1)=1, \\
\chi'(0)=\chi'(1)=0, \quad
\chi''(0)=\chi''(1)=0.\end{gathered}
\]
Those conditions uniquely determine the quintic polynomial in~\eqref{eq:quintic_smoothstep}. Thus the smoothed control reaches the same endpoint values as the scheduled control while matching zero slope and zero curvature at both ends of the transition window. Its time slope is
\[
\frac{d\bu_{\tau_s}}{dt}
=
\frac{\bar{\bu}_{k+1}-\bar{\bu}_k}{\tau_s}\,\chi'(\alpha_k(t)),
\]
where smaller $\tau_s$ results in a steeper transition while preserving the same endpoint values. The same componentwise formula is used for pump speeds, valve openings, and the relaxed mode weights that represent open/active/shut valve logic in our implementations. The mode weights are continuous variables in $[0,1]$ that replace the discrete open, active, and shut valve logic during a transition. Fig.~\ref{fig:smooth_control_regularization} shows the smoothing considered herein. The extended period simulations therefore use the smoothed DAE
\begin{equation}\label{eq:smoothed_compact_dae}
\dot{\bz}_{\tau_s}=\bfm(\bx_{\tau_s},\bu_{\tau_s},\bd),\qquad
\mathbf{0}=\bh(\bx_{\tau_s},\bu_{\tau_s},\bd),
\end{equation}
where $\bx_{\tau_s}$ denotes the trajectory generated by the regularized control path. The reference trajectory $\bx$ is the piecewise smooth trajectory obtained by solving the fixed mode DAE between update instants and restarting the residual after each update. The comparison is therefore meaningful only away from the update instants, where the instantaneous schedule is smooth inside each interval but generally not differentiable at the event itself.

\begin{theorem}[Convergence of the smoothed DAE]\label{thm:smoothed_switching}
Assume that the quintic interpolation in~\eqref{eq:smooth_control_interpolation} remains inside the neighborhood where Assumption~\ref{ass:regular_fixed_mode} holds for every smoothing window. Then the smoothed DAE~\eqref{eq:smoothed_compact_dae} is locally of index two and admits a unique local solution on each such window. Furthermore, for any compact subinterval $\mathcal{T}_{\mathrm{c}}\subset(t_k,t_{k+1})$ whose distance from the switching instants is positive, $\bx_{\tau_s}$ converges uniformly to the corresponding piecewise smooth instantaneous switch trajectory $\bx$ on $\mathcal{T}_{\mathrm{c}}$ as $\tau_s$ tends to zero.
\end{theorem}
The proof is given in Appendix~\ref{app:proof_smoothed_switching}. Theorem~\ref{thm:smoothed_switching} quantifies the difference between the smoothed trajectory and the instantaneous switch trajectory. This convergence statement is deliberately local in time. For any finite smoothing width, the smoothed and instantaneous switch trajectories are not identical, and the largest discrepancy is expected near the transition windows. The theorem only says that this discrepancy disappears on compact intervals strictly between switching instants as the smoothing width is made shorter. This justifies using the smoothed DAE as the local model, since the subsequent analyses require a differentiable residual and fixed pump and valve modes.
\begin{remark}[Smoothing parameter $\tau_s$]
The smoothing width $\tau_s$ is a tuning parameter. If $\tau_s$ is too small, the residual remains very steep and the numerical advantage of smoothing is lost. If $\tau_s$ is too large, the simulated controls no longer resemble the intended and actual control schedules. 
\end{remark}
\subsection{Stability and Controllability of the Linearized Hydraulic DAE}\label{sec:graph_theory}
The preceding section develops the nonlinear hydraulic DAE in its full state form and introduces the smoothing procedure. This section extracts the local graph structure hidden in that model and quantifies what can be inferred from it. We proceed in two steps. First, we show that the local linearization has a DAE structure governed by a weighted Laplacian that makes the network interconnection explicit. Second, we quantify how pipe parameter variations perturb the local linearization, its stability and controllability margins.

\subsubsection{Weighted Laplacian Structure and Local Stability}\label{sec:graph_structure}
This section rewrites the smoothed DAE with fixed controls so that the underlying network graph topology appears explicitly in the linearized state-space matrices. Throughout this section, we consider the smoothed hydraulics defined earlier on a fixed control regime. Fix the positive flow directions and define the signed incidence matrices $
\bN_J\in\mathbb{R}^{n_J\times n_E}, 
\bN_A\in\mathbb{R}^{n_A\times n_E},$ and $
\bN_R\in\mathbb{R}^{n_R\times n_E},$ where $[\bN]_{ke}=+1$ if link $e$ leaves node $k$ and $-1$ if it enters. This is the same outflow minus inflow convention as the incidence matrices in Section~\ref{sec:methods}, but now written on the unified active link ordering. The full incidence matrix is $\bN=\{\bN_J,\bN_A,\bN_R\}\in\mathbb{R}^{(n_J+n_A+n_R)\times n_E}$. 
For the signed link flow vector $\bq\in\mathbb{R}^{n_E}$, define $\boldsymbol{\eta}$ as the unified signed link equation
\[
\boldsymbol{\eta}(\bq,\bu)=\{\boldsymbol{\varphi}^{P}(\bq^{P}),-\boldsymbol{\psi}^{M}(\bq^{M},\bs),\boldsymbol{\varphi}^{W}(\bq^{W},\bo)\}.
\]
Here $\boldsymbol{\eta}$ stacks head losses for dissipative links and the negative of the pump head gain for pumps. The flow and pressure drops through valves are absorbed into the smoothed headloss functions $\boldsymbol{\psi}^{M}$ and $\boldsymbol{\varphi}^{W}$ because the discrete modes are now approximated by continuous interpolations. Positive demand components represent water leaving the network. The symbols $\widetilde{\bd}^{J}$ and $\widetilde{\bd}^{A}$ are signed demand vectors in the incidence convention above, that is $\widetilde{\bd}^{J}=-\bd^{J}$ and $\widetilde{\bd}^{A}=-\bd^{A}$ (i.e., they are not new demand profiles). They only allow the graph DAE to use one incidence sign convention in every nodal row. The fixed control, smoothed DAE can therefore be rewritten as 
\begin{subequations}\label{eq:graph_form_dae}
\begin{align}
\boldsymbol{\Lambda}_q\dot{\bq} &= \bN_J^\top\bp^J+\bN_A^\top\bp^A+\bN_R^\top\bp^R-\boldsymbol{\eta}(\bq,\bu), \label{eq:graph_form_dae_links}\\
\mathbf{0} &= \bN_J\bq-\widetilde{\bd}^{J}, \label{eq:graph_form_dae_junctions}\\
\bA^A\dot{\bp}^{A} &= -\bN_A\bq+\widetilde{\bd}^{A}, \label{eq:graph_form_dae_tanks}\\
\mathbf{0} &= \bp^R-\bar{\bp}^{R}. \label{eq:graph_form_dae_reservoirs}
\end{align}
\end{subequations}
Equation~\eqref{eq:graph_form_dae} represents the full DAE in~\eqref{eq:full_hydraulic_dae}, rewritten by stacking all active links into a unified flow vector $\bq$ and embedding all nonreservoir heads into a standard graph incidence convention. Thus $\bN$ captures the topology, while the nonlinear hydraulic physics are isolated in $\boldsymbol{\eta}$, $\boldsymbol{\Lambda}_q$, and $\bA^A$. We now consider an equilibrium $(\bar{\bq},\bar{\bp},\bar{\bu},\bar{\bd})$ of~\eqref{eq:graph_form_dae}. 
\begin{assumption}[Positive link slopes]\label{ass:incremental_slopes}
Every active link equation is differentiable at that equilibrium and has positive incremental slope
\begin{equation}~\label{eq:kappae}
\kappa_e=\frac{\partial \eta_e}{\partial q_e}(\bar{\bq},\bar{\bu})>0,
\qquad e\in\mathcal{E}.
\end{equation}
\end{assumption}
Assumption~\ref{ass:incremental_slopes} is local and can be interpreted as follows. It states that, near the operating point, pushing more flow through a pipe or open valve increases the headloss, while pushing more flow through a pump decreases its available head gain. The slope $\kappa_e$ is thus the change in head loss across link $e$ per unit change in its flow, that is, the local hydraulic resistance of the link, and its inverse $\kappa_e^{-1}$ is the local hydraulic conductance. Such strict monotonicity naturally parameterizes the resulting linearized network dynamics via a weighted graph Laplacian. The Laplacian introduced in Theorem~\ref{thm:graph_linear_dae} differs from the topological Laplacians used in graph spectral studies of WDNs~\cite{dinardo2018applications,yazdani2011complex}. Its weights $w_e = \kappa_e^{-1}$ are not binary or based on length; they are the incremental hydraulic conductances at the operating point.

Let the local deviations be $\Delta\bq(t)=\bq(t)-\bar{\bq}$, $\Delta\bp^\alpha(t)=\bp^\alpha(t)-\bar{\bp}^\alpha$ for $\alpha\in\{J,A,R\}$, and $\Delta\bx_h(t)=\{\Delta\bq(t),\Delta\bp^{J}(t),\Delta\bp^{A}(t)\}.$ The subscript $h$ denotes the hydraulic graph state used in this section. It excludes the fixed reservoir heads and the discrete valve mode variables that are replaced by continuously differentiable smoothing parameters embedded in $\bar{\bu}$. Linearizing~\eqref{eq:graph_form_dae} at the equilibrium $(\bar{\bq},\bar{\bp},\bar{\bu},\bar{\bd})$ gives the linear DAE state-space model
\begin{equation}\label{eq:linearized_graph_dae}
\bE_h\,\Delta\dot{\bx}_h(t)
=
\bA_h\Delta\bx_h(t)+\bB_u\Delta\bu(t)+\bB_d\Delta\bd(t),
\end{equation}
where the state-space matrices are constructed from the incidence matrices and the local link slopes $\bK=\operatorname{diag}(\kappa_e)$ as follows
\begin{equation}\label{eq:linear_dae_blocks}
	\resizebox{\columnwidth}{!}{$
		\displaystyle
		\bE_h=
		\begin{bmatrix}
			\boldsymbol{\Lambda}_q & \mathbf{0} & \mathbf{0}\\
			\mathbf{0} & \mathbf{0} & \mathbf{0}\\
			\mathbf{0} & \mathbf{0} & \bA^A
		\end{bmatrix},
		\qquad
		\bA_h=
		\begin{bmatrix}
			-\bK & \bN_J^\top & \bN_A^\top\\
			-\bN_J & \mathbf{0} & \mathbf{0}\\
			-\bN_A & \mathbf{0} & \mathbf{0}
		\end{bmatrix}.
		$}
\end{equation}
Appendix~\ref{app:threenodes} derives the nonlinear and the linearized DAE of the three node network.
Let $\mathcal{G}_{\mathrm{act}}=(\mathcal{V},\mathcal{E}_{\mathrm{act}})$ denote the graph induced by the links whose differentiable pipe, pump, or valve equation is included at the fixed operating point. For a time varying head perturbation $\Delta\bp(t)$ and the induced link flow perturbation $\Delta\bq(t)$, define the signed nodal imbalance $
\Delta\boldsymbol{\iota}(t):=\bN\Delta\bq(t).$

We use the full matrix $\bN$, rather than only $\bN_J$, because the Laplacian relation acts on junction, tank, and reservoir nodes together. The vector $\Delta\boldsymbol{\iota}(t)$ is not an imposed demand perturbation. It is the signed nodal flow imbalance generated by a link flow deviation. If the algebraic junction rows and the tank storage rows are not enforced, a nonzero value moves the DAE away from its balanced operating point. We next demonstrate that the linearized DAE is asymptotically stable on the set of states satisfying the linearized algebraic constraints, provided every hydraulically active subnetwork is anchored to a reservoir. 
\begin{theorem}[Local stability analysis through the weighted Laplacian]\label{thm:graph_linear_dae}
Under Assumption~\ref{ass:incremental_slopes}, define the local slope matrix $\bK=\operatorname{diag}(\kappa_e)$. Then $\bK\succ0$ and the conductance matrix $\bW=\bK^{-1}\succ0$ forms the weighted graph Laplacian
\begin{equation}\label{eq:weighted_laplacian}
\bL_w=\bN\bW\bN^\top
\end{equation}
representing the linearized hydraulic conductance of the network, ensuring that small signed nodal imbalances satisfy $\Delta\boldsymbol{\iota}(t)=\bL_w\Delta\bp(t)$.  Furthermore, defining a Lyapunov function as the incremental hydraulic energy $V(\Delta\bx_h(t)) = \frac{1}{2}\Delta\bx_h^\top(t)\bE_h\Delta\bx_h(t) \geq 0$, any consistent unforced trajectory of the linearized DAE satisfying the algebraic constraints obeys the energy dissipation identity
\begin{equation}\label{eq:linear_dae_energy}
\dot V(\Delta\bx_h(t)) =-\Delta\bq^\top(t)\bK\Delta\bq(t) \leq 0.
\end{equation} 
Consequently, if every connected component of the active link graph $\mathcal{G}_{\mathrm{act}}$ contains at least one reservoir node, the equilibrium is locally asymptotically stable.
\end{theorem}
The proof is given in Appendix~\ref{app:proof_graph_linear_dae}. It details the connection between the Laplacian structure, the energy dissipation rate, and the reservoir connectivity requirement. This local result indicates the pump and valve modes and the operating point for which Assumption~\ref{ass:incremental_slopes} holds. Stability of this DAE means stability of all finite DAE modes (the eigenvalues of the pencil $(\bA_h,\bE_h)$ that are not at infinity) for consistent initial conditions, with the algebraic junction and reservoir constraints enforced throughout the motion. This differs from an unconstrained ODE statement because the admissible perturbations must remain on the linearized constraint set. The pencil $(\bA_h,\bE_h)$ has $n_q+n_A-n_J$ finite eigenvalues. The remaining eigenvalues are at infinity, with one Jordan block of size two per junction, which is the index two structure of Theorem~\ref{thm:local_index_one} in linear form.

The Lyapunov function $V$ is the incremental hydraulic energy, i.e., the sum of kinetic energy in the links and potential energy in the tanks. Along unforced trajectories it decays monotonically because the network headloss dissipates energy. The weighted Laplacian enters through the static part of the linearization. The relation $\bN^\top\Delta\bp=\bK\Delta\bq$ gives $\Delta\boldsymbol{\iota}=\bN\bK^{-1}\bN^\top\Delta\bp=\bL_w\Delta\bp$. The derivative $\dot V$ is written with $\bK$ rather than $\bL_w$ because it measures the same dissipation in link flow coordinates. When the controls are fixed, the local monotonicity imposed by Assumption~\ref{ass:incremental_slopes} guarantees that the network dissipates energy and returns to the nominal equilibrium following small disturbances.

\subsubsection{Accuracy of the Linearization Under Parameter Variations}\label{sec:linearization_accuracy}
The objective of this section is to quantify when the nominal linear DAE in~\eqref{eq:linearized_graph_dae} remains accurate to the prescribed tolerance, meaning that it predicts the first-order shift of both the fixed mode equilibrium and the graph state matrix within a prescribed local error tolerance, after hydraulic parameters are perturbed. The practical question is straightforward. Suppose a pipe roughness or pump coefficient is not known exactly. How far may that parameter move before the nominal linearization should no longer be used for local analysis? To answer this, let $\boldsymbol{\vartheta}\in\mathbb{R}^{n_\vartheta}$ denote the vector of positive nominal values for uncertain hydraulic parameters in the nonlinear fixed control DAE~\eqref{eq:graph_form_dae}, such as pipe roughnesses, pipe diameters, valve coefficients, pump curve coefficients, and tank areas. These could be a subset of all WDN parameters. Fix the controls and demands at $(\bar{\bu},\bar{\bd})$ in the nonlinear residual, solve the corresponding equilibrium, and then form the linearized matrix in~\eqref{eq:linearized_graph_dae}. Define the equilibrium residual as a function of the parameters
\begin{equation}\label{eq:equilibrium_map}
\mathbf{r}_{\vartheta}(\bx,\boldsymbol{\vartheta})
=
\{\bfm(\bx,\bar{\bu},\bar{\bd};\boldsymbol{\vartheta}),\bh(\bx,\bar{\bu},\bar{\bd};\boldsymbol{\vartheta})\}.
\end{equation}
An equilibrium therefore satisfies $\mathbf{r}_{\vartheta}(\bar{\bx},\bar{\boldsymbol{\vartheta}})=\mathbf{0}$. Define the nominal Jacobians $\bJ_x=\partial_{\bx}\mathbf{r}_{\vartheta}(\bar{\bx},\bar{\boldsymbol{\vartheta}})$ and $
\bJ_\vartheta=\partial_{\boldsymbol{\vartheta}}\mathbf{r}_{\vartheta}(\bar{\bx},\bar{\boldsymbol{\vartheta}}).$
When $\bJ_x$ is nonsingular, the first-order equilibrium sensitivity is
\begin{equation}\label{eq:equilibrium_sensitivity}
\bS_\vartheta=-\bJ_x^{-1}\bJ_\vartheta .
\end{equation}
Let $\bx^\star(\boldsymbol{\vartheta})$ denote the local equilibrium map of the full fixed mode state $\bx$ supplied by the implicit function theorem, so $\mathbf{r}_{\vartheta}(\bx^\star(\boldsymbol{\vartheta}),\boldsymbol{\vartheta})=\mathbf{0}$ near $\bar{\boldsymbol{\vartheta}}$. Let $\bA_h(\boldsymbol{\vartheta})$ denote the state matrix in the graph linearized DAE~\eqref{eq:linearized_graph_dae}-\eqref{eq:linear_dae_blocks}, evaluated at the equilibrium associated with $\boldsymbol{\vartheta}$. In particular, the nominal matrix in Theorem~\ref{thm:graph_linear_dae} is $\bA_h(\bar{\boldsymbol{\vartheta}})$. The next theorem compares the exact perturbed equilibrium and the exact perturbed graph state matrix $\bA_h(\boldsymbol{\vartheta})$ with their first-order predictions.

\begin{theorem}[Quadratic error bounds for parameter perturbations]\label{thm:linearization_bounds}
Suppose $\mathbf{r}_{\vartheta}(\cdot)$ is three times continuously differentiable near $(\bar{\bx},\bar{\boldsymbol{\vartheta}})$ and $\bJ_x$ is nonsingular. Then there exist $r_0,L_x,L_A>0$ and an equilibrium map $\bx^\star(\boldsymbol{\vartheta})$ for $\|\boldsymbol{\vartheta}-\bar{\boldsymbol{\vartheta}}\|<r_0$ such that, for every $\Delta\boldsymbol{\vartheta}$ in this ball,
\begin{equation}\label{eq:equilibrium_shift_bound}
\Big\|
\underbrace{\bx^\star(\bar{\boldsymbol{\vartheta}}+\Delta\boldsymbol{\vartheta})-\bar{\bx}}_{\text{true equilibrium shift}}
-
\underbrace{\bS_\vartheta\Delta\boldsymbol{\vartheta}}_{\text{first-order prediction}}
\Big\|
\le
\frac{L_x}{2}\,\|\Delta\boldsymbol{\vartheta}\|^2
\end{equation}
and
\begin{equation}\label{eq:linearization_matrix_bound}
\Big\|
\underbrace{\bA_h(\bar{\boldsymbol{\vartheta}}+\Delta\boldsymbol{\vartheta})-\bA_h(\bar{\boldsymbol{\vartheta}})}_{\text{true matrix shift}}
-
\underbrace{\nabla_{\boldsymbol{\vartheta}}\bA_h(\bar{\boldsymbol{\vartheta}})\Delta\boldsymbol{\vartheta}}_{\text{first-order prediction}}
\Big\|
\le
\frac{L_A}{2}\,\|\Delta\boldsymbol{\vartheta}\|^2.
\end{equation}
Consequently, a prescribed relative matrix error tolerance $\varepsilon_{\mathrm{lin}}\in(0,1)$ is guaranteed whenever
\begin{equation}\label{eq:linearization_validity_radius}
\|\Delta\boldsymbol{\vartheta}\|
\le
r_{\mathrm{lin}}
=
\min\left\{
r_0,\,
\sqrt{\frac{2\varepsilon_{\mathrm{lin}}\|\bA_h(\bar{\boldsymbol{\vartheta}})\|}{L_A}}
\right\}.
\end{equation}
\end{theorem}
The proof is given in Appendix~\ref{app:proof_linearization_bounds}. Theorem~\ref{thm:linearization_bounds} assumes that $\mathbf{r}_{\vartheta}$ is three times continuously differentiable. The error that arises when the perturbed $\bA_h$ is estimated from its nominal value is bounded according to its sensitivity to the parameters by a term proportional to $\|\Delta\boldsymbol{\vartheta}\|^2$. This bound requires two derivatives of $\bA_h$, and therefore three derivatives of $\mathbf{r}_{\vartheta}$. The Hazen-Williams and Darcy-Weisbach equations and the polynomial pump curves in~\eqref{eq:component_nonlinearities} satisfy this condition at any operating point with nonzero link flows, since they are smooth away from zero flow. This theorem therefore establishes an explicit validity radius bounding the shift in both the equilibrium and the linearized state matrix caused by uncertain hydraulic parameters. Intuitively, the linear sensitivity matrix $\bS_\vartheta$ predicts the leading-order shifts, while the quadratic bounds~\eqref{eq:equilibrium_shift_bound} and~\eqref{eq:linearization_matrix_bound} quantify the error introduced by discarding higher-order nonlinearities. The theoretical radius derived in~\eqref{eq:linearization_validity_radius} guarantees that as long as the parameter deviations are smaller than $r_{\mathrm{lin}}$, the linear model remains faithful to the prescribed tolerance $\varepsilon_{\mathrm{lin}}$. 

Algorithm~\ref{proc:linearization_screen} checks whether, after a roughness change, the state matrix can still be estimated from its nominal value and its sensitivity to roughness, without solving the network equations again at the new roughness. The radius $r_{\mathrm{lin}}$ in~\eqref{eq:linearization_validity_radius} is the largest parameter change where $\bA_h(\bar{\boldsymbol{\vartheta}})+\nabla_{\boldsymbol{\vartheta}}\bA_h(\bar{\boldsymbol{\vartheta}})\Delta\boldsymbol{\vartheta}$, stays within the tolerance $\varepsilon_{\mathrm{lin}}$. The constant $L_A$ measures how quickly the sensitivity of $\bA_h$ to the parameters changes as the parameters change. In practice, it is estimated from a few evaluations of $\bA_h$ at nearby parameter values, which describes its behavior near $\bar{\boldsymbol{\vartheta}}$ but does not guarantee a bound for every parameter change within the radius. Consequently, the numerical radius is an estimate unless $L_A$ is independently bounded on that neighborhood.

\SetAlgorithmName{Algorithm}{Algorithm}{List of Algorithms}
\begin{algorithm}[t]
\caption{Reusing a linearized hydraulic model after a change in pipe roughness}
\label{proc:linearization_screen}
\DontPrintSemicolon
\footnotesize
\textbf{input:} network file, operating point (demands, pump speeds, valve settings), nominal roughness $\bar{\boldsymbol{\vartheta}}$, proposed roughness change $\Delta\boldsymbol{\vartheta}$ (from calibration, aging, or a scenario study), relative error tolerance $\varepsilon_{\mathrm{lin}}$\;
\textbf{solve:} the fixed control equilibrium $\bar{\bx}$ at $\bar{\boldsymbol{\vartheta}}$ and form $\bE_h$, $\bA_h(\bar{\boldsymbol{\vartheta}})$ from~\eqref{eq:linear_dae_blocks} // \textit{one hydraulic solve}\;
\textbf{compute:} the equilibrium sensitivity $\bS_\vartheta$ from~\eqref{eq:equilibrium_sensitivity} // \textit{how heads and flows move with roughness}\;
\textbf{estimate:} the curvature constant $L_A$ from two finite difference evaluations of $\bA_h(\boldsymbol{\vartheta})$ along $\Delta\boldsymbol{\vartheta}$\;
\textbf{compute:} the validity radius $r_{\mathrm{lin}}$ from~\eqref{eq:linearization_validity_radius}\;
\eIf{$\|\Delta\boldsymbol{\vartheta}\|\le r_{\mathrm{lin}}$}{
\textbf{keep:} $\bA_h(\bar{\boldsymbol{\vartheta}})$ for the stability and controllability analysis; \textbf{report:} the equilibrium and matrix errors from~\eqref{eq:equilibrium_shift_bound}-\eqref{eq:linearization_matrix_bound}\;
}{
\textbf{relinearize:} solve the equilibrium at $\bar{\boldsymbol{\vartheta}}+\Delta\boldsymbol{\vartheta}$ and form $\bA_h$ again // \textit{the old model is no longer accurate to $\varepsilon_{\mathrm{lin}}$}\;
}
\textbf{output:} the matrix to use for the analysis and the guaranteed error of the first-order prediction
\end{algorithm}

\subsubsection{Stability and Controllability Margins Under Hydraulic Parameter Uncertainty}\label{sec:linear_dae_properties}
The objective of this section is to analyze the local stability and controllability properties of the linearized hydraulic network, and to quantify how robust these properties are to parametric perturbations. Throughout this analysis, we consider the reduced, smooth, linearized DAE. Throughout the paper, the word reduced refers to the local elimination of the algebraic states, nonlinear in~\eqref{eq:local_reduced_ode} and linear in~\eqref{eq:reduced_linear_system}. It does not refer to model order reduction, and none of the differential states are eliminated.
Once the smoothed DAE is locally reduced via Theorem~\ref{thm:local_index_one}, the minimal differential deviation is restricted to $\Delta\bz=\{\Delta\bq,\Delta\bp^A\}$. Based on the linearized DAE in~\eqref{eq:linearized_graph_dae}, the junction head perturbations $\Delta\bp^J$ are eliminated through the linearized form of the hidden constraint~\eqref{eq:hidden_constraint}, and the flow perturbations are restricted to the set $\bN_J\Delta\bq=\mathbf{0}$ of the linearized junction balance. The reduced state therefore has $n_q+n_A-n_J$ components. The resulting reduced continuous-time state-space matrices are $\bA_{\mathrm{red}}(\boldsymbol{\vartheta})$, $\bB_{\mathrm{red}}(\boldsymbol{\vartheta})$, and $\bD_{\mathrm{red}}(\boldsymbol{\vartheta})$, with
\begin{equation}\label{eq:reduced_linear_system}
\Delta\dot{\bz}
=
\bA_{\mathrm{red}}(\boldsymbol{\vartheta})\,\Delta\bz
+\bB_{\mathrm{red}}(\boldsymbol{\vartheta})\,\Delta\bu
+\bD_{\mathrm{red}}(\boldsymbol{\vartheta})\,\Delta\bd.
\end{equation}
At the nominal parameter vector, suppose $\bA_{\mathrm{red}}(\bar{\boldsymbol{\vartheta}})=\bV\boldsymbol{\Lambda}\bV^{-1}$ and define the stability margin
$\alpha_s=-\max_i \Re \lambda_i(\boldsymbol{\Lambda})>0.$ Here, $\alpha_s$ is the slowest rate where pressures, flows, and tank levels return to their steady values after a disturbance at the current operating point. Thus $1/\alpha_s$ is the longest recovery time constant of the network. It is governed by the tank storage modes in the case studies, while the pipe inertia sets the faster recovery of pressures and flows discussed in Section~\ref{sec:case_theorem3}.
For a horizon $N\ge 1$, define the controllability matrix and its smallest singular value by
\begin{equation}\label{eq:kalman_matrix_def}
\begin{split}
\boldsymbol{\mathcal{K}}_N(\boldsymbol{\vartheta})&=\big[\bB_{\mathrm{red}}(\boldsymbol{\vartheta}),\ \bA_{\mathrm{red}}(\boldsymbol{\vartheta})\bB_{\mathrm{red}}(\boldsymbol{\vartheta}),\\
&\qquad\ldots,\ \bA_{\mathrm{red}}(\boldsymbol{\vartheta})^{N-1}\bB_{\mathrm{red}}(\boldsymbol{\vartheta})\big],
\end{split}
\end{equation}
\begin{equation}\label{eq:controllability_margin_def}
\sigma_c=\sigma_{\min}\big(\boldsymbol{\mathcal{K}}_N(\bar{\boldsymbol{\vartheta}})\big)>0.
\end{equation}
The controllability margin $\sigma_c$ measures how much the pump speeds and valve settings can move the hydraulic state over time in the direction that is hardest to move. A small $\sigma_c$ means that some combination of pressures and tank levels responds weakly to every available pump and valve actuation.
Finally, let $c_A$ and $c_B$ be local constants satisfying
\begin{subequations}\label{eq:state_space_perturbation_bound}
\begin{align}
\|\bA_{\mathrm{red}}(\bar{\boldsymbol{\vartheta}}+\Delta\boldsymbol{\vartheta})-\bA_{\mathrm{red}}(\bar{\boldsymbol{\vartheta}})\|
&\le c_A\|\Delta\boldsymbol{\vartheta}\|, \label{eq:state_space_perturbation_bound_A}\\
\|\bB_{\mathrm{red}}(\bar{\boldsymbol{\vartheta}}+\Delta\boldsymbol{\vartheta})-\bB_{\mathrm{red}}(\bar{\boldsymbol{\vartheta}})\|
&\le c_B\|\Delta\boldsymbol{\vartheta}\|, \label{eq:state_space_perturbation_bound_B}
\end{align}
\end{subequations}
and let $c_{\mathcal{K}}$ satisfy
\begin{equation}\label{eq:kalman_perturbation_bound}
\big\|\boldsymbol{\mathcal{K}}_N(\bar{\boldsymbol{\vartheta}}+\Delta\boldsymbol{\vartheta})-\boldsymbol{\mathcal{K}}_N(\bar{\boldsymbol{\vartheta}})\big\|
\le
c_{\mathcal{K}}(c_A+c_B)\|\Delta\boldsymbol{\vartheta}\|.
\end{equation}
Such constants exist locally whenever the reduced matrices are continuously differentiable functions of $\boldsymbol{\vartheta}$.
Here $\kappa(\bV)=\|\bV\|\,\|\bV^{-1}\|$ is the condition number of the eigenvector matrix, $\alpha_s$ is the nominal exponential stability margin, $\sigma_c$ is the nominal finite-horizon controllability margin, and $\sigma_{\min}(\cdot)$ is the smallest singular value operator.

\begin{algorithm}[t]
\caption{Ranking pipes by their effect on the stability and controllability margins}
\label{proc:margin_screen}
\DontPrintSemicolon
\footnotesize
\textbf{input:} linearized model $(\bE_h,\bA_h,\bB_u)$ at the operating point, horizon $N$, set of pipes $\mathcal{P}$, roughness change $\delta\theta_i$ to test on each pipe, finite difference step $h_\theta$\;
\textbf{compute:} the stability margin $\alpha_s$, the controllability margin $\sigma_c$ from~\eqref{eq:kalman_matrix_def}-\eqref{eq:controllability_margin_def}, and the eigenvector conditioning $\kappa(\bV)$\;
\For{each pipe $i\in\mathcal{P}$}{
\textbf{perturb:} the roughness of pipe $i$ by $\pm h_\theta$, keep the operating point, and form $\bA_h$ and $\boldsymbol{\mathcal{K}}_N$ for both perturbations\;
\textbf{compute:} the sensitivities $g_{A,i}=\|\partial\bA_h/\partial\theta_i\|_2$ and $g_{K,i}=\|\partial\boldsymbol{\mathcal{K}}_N/\partial\theta_i\|_2$ by central differences // \textit{effect of pipe $i$ on recovery and on pump reach}\;
\textbf{compute:} the guaranteed remaining margins after the change $\delta\theta_i$, the first-order form of~\eqref{eq:stability_margin_perturbation}-\eqref{eq:controllability_margin_perturbation},
$\widehat{\alpha}_{s,i}=\alpha_s-\kappa(\bV)g_{A,i}|\delta\theta_i|$ and $\widehat{\sigma}_{c,i}=\sigma_c-g_{K,i}|\delta\theta_i|$\;
}
\textbf{output:} pipes ranked by decreasing $g_{A,i}$ (largest effect on how fast pressures and flows recover) and by decreasing $g_{K,i}$ (largest effect on what the pumps can reach), with $\widehat{\alpha}_{s,i}$ and $\widehat{\sigma}_{c,i}$ for each pipe
\end{algorithm}

\begin{theorem}[Stability and controllability analysis]\label{thm:property_robustness}
Assume that $\bA_{\mathrm{red}}(\bar{\boldsymbol{\vartheta}})$ is diagonalizable, so that $\bA_{\mathrm{red}}(\bar{\boldsymbol{\vartheta}})=\bV\boldsymbol{\Lambda}\bV^{-1}$ with $\kappa(\bV)<\infty$. A sufficiently small perturbation $\Delta\boldsymbol{\vartheta}$ preserves the following properties.
\begin{enumerate}
\item The reduced model \eqref{eq:reduced_linear_system} remains exponentially stable whenever $\kappa(\bV)c_A\|\Delta\boldsymbol{\vartheta}\|<\alpha_s$.
\item The reduced model \eqref{eq:reduced_linear_system} remains controllable over horizon $N$ whenever $c_{\mathcal{K}}(c_A+c_B)\|\Delta\boldsymbol{\vartheta}\|<\sigma_c$.
\item The margin changes obey
\begin{subequations}\label{eq:margin_perturbation_pair}
\begin{align}
\big|\alpha_s(\bar{\boldsymbol{\vartheta}}+\Delta\boldsymbol{\vartheta})-\alpha_s\big|
&\le
\kappa(\bV)c_A\|\Delta\boldsymbol{\vartheta}\|+O(\|\Delta\boldsymbol{\vartheta}\|^2), \label{eq:stability_margin_perturbation}\\
\sigma_{\min}\big(\boldsymbol{\mathcal{K}}_N(\bar{\boldsymbol{\vartheta}}+\Delta\boldsymbol{\vartheta})\big)
&\ge
\sigma_c-c_{\mathcal{K}}(c_A+c_B)\|\Delta\boldsymbol{\vartheta}\|. \label{eq:controllability_margin_perturbation}
\end{align}
\end{subequations}
\end{enumerate}
\end{theorem}
The proof is given in Appendix~\ref{app:proof_property_robustness}. Theorem~\ref{thm:property_robustness} allows us to quantify how the hydraulic parameters affect the stability and controllability of the network. The bounds have the virtue of being explicit and computable from a single singular value decomposition and finite difference gradient evaluation. Pipe roughness, pipe diameter, pump coefficients, and valve coefficients influence stability and controllability in two ways. They perturb the local state-space matrices through the same dependence of the linearization on the parameters bounded in Theorem~\ref{thm:linearization_bounds}. They also shrink or enlarge the margins in~\eqref{eq:margin_perturbation_pair}. In that sense, local control properties depend not only on topology, but also on how strongly each hydraulic element couples heads and flows.

\begin{figure*}[t]
	\centering
	\includegraphics[width=\textwidth]{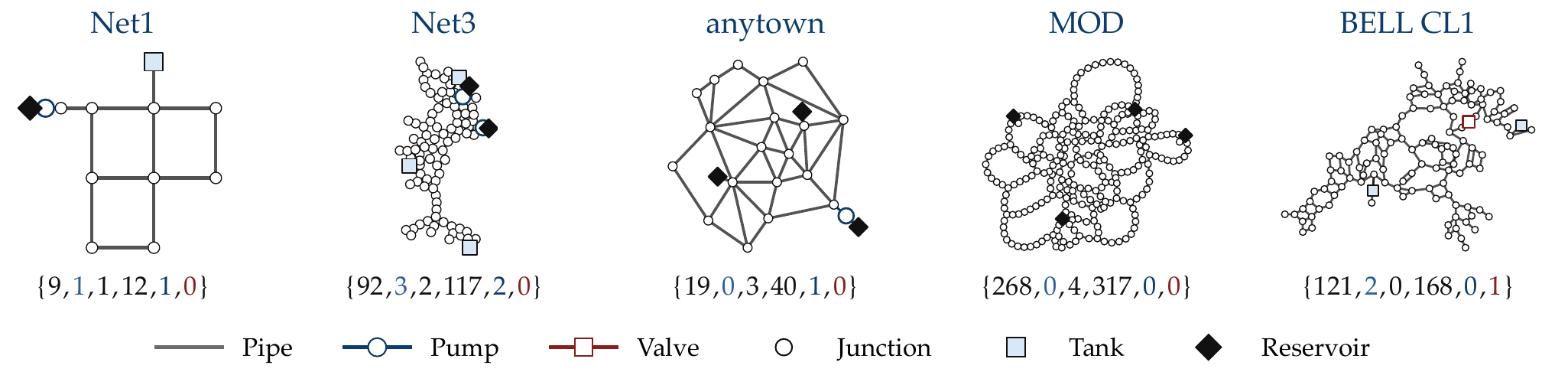}
	\caption{Benchmark networks considered in the case studies. The figure depicts the topologies of the five networks, where the size vector $\{n_J,n_A,n_R,n_{\mathrm{pipe}},n_M,n_W\}$ below each network gives the number of junctions, tanks, reservoirs, pipes, pumps, and valves.}
	\label{fig:network_overview}
\end{figure*}

Algorithm~\ref{proc:margin_screen} reports direct roughness sensitivities at fixed nominal flows, as implemented in the numerical study. The algorithm checks one parameter direction at a time through the local Jacobians of $\bA_h$ and $\boldsymbol{\mathcal{K}}_N$. In this form, the most critical parameter is simply the one that leaves the smallest guaranteed margin after the candidate perturbation is applied. The case studies investigate this point by testing the model responses on networks whose parameters span very different regimes.

\textsc{\textbf{Pipe identification and pump allocation}.} Theorem~\ref{thm:property_robustness} can be useful in practice as follows. \textit{(i)} Algorithm~\ref{proc:margin_screen} computes pipes that affect the recovery of pressures and flows ($g_{A,i}$) and the reach of the pumps ($g_{K,i}$). Note that this ranking follows the dynamics and not the steady flows, since a pipe with a small steady flow can still have a large slope $\kappa_e$ and hence a large effect on $\bA_{\mathrm{red}}$. \textit{(ii)} The controllability matrix~\eqref{eq:kalman_matrix_def} can help identify where an added pump increases the reach of the existing pumps. A pump placed on a link adds one column to $\bB_{\mathrm{red}}$. For a set $\mathcal{S}$ of candidate links, let $\boldsymbol{\mathcal{K}}_N(\mathcal{S})$ be the controllability matrix built from the existing pumps and of the links in $\mathcal{S}$, and let $\bW_N(\mathcal{S})=\boldsymbol{\mathcal{K}}_N(\mathcal{S})\boldsymbol{\mathcal{K}}_N(\mathcal{S})^\top$ be the finite-horizon controllability Gramian. For a small $\epsilon>0$, the set function $\log\det\big(\bW_N(\mathcal{S})+\epsilon\bI\big)$ is monotone submodular~\cite{summers2016submodularity}, meaning that adding a pump does not decrease the objective function; it exhibits diminishing returns as the set size increases. As such, the greedy algorithm~\cite{nemhauser1978analysis} that adds at each step the link
\[
j^\star=\arg\max_{j\notin\mathcal{S}}\;\log\det\big(\bW_N(\mathcal{S}\cup\{j\})+\epsilon\bI\big)
\]
returns a set where its gain over the existing pumps is at least $(1-1/e)\approx 0.63$ times the gain of the best set of the same size. Section~\ref{sec:case_actuators} applies this algorithm to \texttt{Net3} and \texttt{MOD}.

\subsubsection{Computational Analysis}\label{sec:computational_analysis}
The practical results in Theorems~\ref{thm:graph_linear_dae}-\ref{thm:property_robustness} require hydraulic simulation, equilibrium solves, eigenvalue calculations, singular value calculations, and finite difference sensitivities. They do not solve a large-scale control synthesis problem and they do not rely on any optimization per se. The only nonlinear numerical solve is the fixed mode hydraulic equilibrium solve used before linearization. After that, the computations are based on simple matrix computations. Forming $\bE_h$, $\bA_h$, and $\bL_w$ costs $O(n_E)$ operations, where $n_E$ is the number of active links. The generalized eigenvalue computations in Theorem~\ref{thm:graph_linear_dae} costs up to $O(n_x^3)$, where $n_x$ is the local DAE state dimension. The error bounds of Theorem~\ref{thm:linearization_bounds} require $N_{\mathrm{fd}}$ equilibrium solves for the finite difference matrix shifts, each roughly $O(N_{\mathrm{it}}n_x^3)$, where $N_{\mathrm{fd}}$ is the number of finite difference evaluations and $N_{\mathrm{it}}$ the number of nonlinear solver iterations. The margins of Theorem~\ref{thm:property_robustness} require one matrix sensitivity per parameter, that is $O(n_\vartheta)$ evaluations for $n_\vartheta$ parameters, followed by eigenvalue and singular value computations.

\section{Case Study Networks and Numerical Settings}\label{sec:networks_settings}
\subsection{Benchmark Networks}\label{sec:benchmark_networks}
We test six WDNs of different sizes, thereby exposing the theoretical methods and the DAE model to different combinations of storage, pumping, and valve logic. Net1 and Net3 are standard EPANET examples~\cite{rossman2020epanet}. \texttt{Anytown} (design benchmark with one pump and three reservoirs), \texttt{MOD} (Modena, four reservoirs, no pumps or valves), and \texttt{BELL\_CL1} are taken from the KIOS Research EPANET Benchmarks collection~\cite{kios2026benchmarks}. The \texttt{BELL\_CL1} benchmark has one throttle control valve (TCV) and no pumps. Fig.~\ref{fig:network_overview} depicts the topology of the five larger networks and uses the ordered component tuple $\{n_J,n_A,n_R,n_{\mathrm{pipe}},n_M,n_W\}$ for junctions, tanks, reservoirs, pipes, pumps, and valves. Appendix~\ref{app:numerical_settings} lists the numerical parameters used for the case studies; all other network parameters are the EPANET defaults of each input file.

\subsection{Implementation and Numerical Settings}\label{sec:numerical_settings}
All simulations reported in this paper are generated through a Python implementation and use EPyT~\cite{kyriakou2023epyt} as the EPANET interface. The studies are run on an M5 Max MacBook Pro (2026) with 128 GB of RAM. The numerical settings of each study are given in Appendix~\ref{app:numerical_settings}.

\section{Results and Discussion}\label{sec:case_studies}
The case studies investigate the validity and usefulness of the developed theoretical formulations summarized in Fig.~\ref{fig:theory_roadmap_flow}. To that end, we investigate the following research questions.
\begin{itemize}
	\item $(\mathrm{Q}1)$ Does the DAE reproduce the EPANET heads and flows under fixed controls over one hour and under the smoothed pump and valve schedule over 24 hours? Where do the two models differ?
	\item $(\mathrm{Q}2)$ What are the links that govern how quickly pressures and flows recover after a disturbance at a fixed operating point?
	\item $(\mathrm{Q}3)$ For what range of pipe roughness uncertainty does the linearized DAE remain accurate? What are the pipes that affect the local stability and controllability matrices?
	\item $(\mathrm{Q}4)$ How does demand variation change the stability margin and the influence of the pumps on pressures, flows, and tank levels?
	\item $(\mathrm{Q}5)$ As compared to topological and random pump allocation, where does an added pump have the most increase on network controllability?
\end{itemize}
The subsequent sections answer the above research questions while explaining the practical aspects of the proposed DAE representation of WDN hydraulics.

\begin{figure*}[t]
	\centering
	\includegraphics[width=\textwidth]{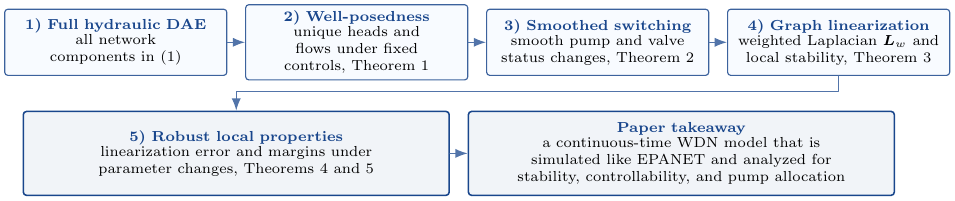}
	\caption{Summary of the paper's technical contributions. Section~\ref{sec:case_studies} validates the proposed WDN model and the theoretical results.}
	\label{fig:theory_roadmap_flow}
\end{figure*}

\subsection{Unique Hydraulics Under Fixed Controls}\label{sec:case_theorem1}
The theoretical methods are not very useful if the DAE model produces trajectories that are substantially different from EPANET's. To that end, this section tests the validity of the DAE model for one time step with fixed controls (Theorem~\ref{thm:local_index_one}). The one step runs use~\eqref{eq:compact_dae_form}. We ensure a fair comparison by matching the component parameters between EPANET and the DAE solver. The only time varying input for the single time step is a smooth positive demand multiplier. Daily water demands are commonly represented through deterministic patterns or periodic demand multipliers in EPANET style simulation, and spatially variable demand patterns are known to affect both design and operation~\cite{rossman2020epanet,walski2003advanced,diao2019demand}. The implementation parameters for this multiplier and the other numerical studies are reported in Appendix~\ref{app:numerical_settings}.

Tab.~\ref{tab:epanet_dae_error_suite} and Fig.~\ref{fig:epanet_comparison}(a) support the following observation. Theorem~\ref{thm:local_index_one} states that the DAE with fixed controls has a unique local solution. Matching trajectories against EPANET cannot \emph{prove} uniqueness, but it does verify that the trajectory selected by the implicit Euler step results in the same unique local solution predicted by the theorem. We note that the implicit Euler step converges for Hessenberg DAEs of index two, with first-order accuracy in both the differential and the algebraic variables~\cite[Thm.~3.2.2]{Brenan1996}. As such, the comparison tests the DAE model and not the numerical integration. Across the six networks, the largest junction head error of the one hour simulations is $0.80$~m for \texttt{BELL\_CL1} and $0.53$~m for \texttt{MOD}. The error remains below $0.23$~m for the other networks (Tab.~\ref{tab:epanet_dae_error_suite}). In these simulations, the device modes are fixed and the DAE is integrated within the neighborhood where Assumption~\ref{ass:regular_fixed_mode} holds. The results therefore demonstrate acceptable performance for the DAE simulation over a single time step under varying demand and fixed controls.

\subsection{Smoothed Control Validation and a Fast Pump Trip}\label{sec:case_theorem3}

This section tests Theorem~\ref{thm:smoothed_switching} for the same networks over a 24 hour horizon with the smoothed control law in~\eqref{eq:smooth_control_interpolation}. The implementation first extracts the resolved EPANET schedule at $300$~s intervals and then applies the same quintic interpolation to the control channels of every link: pump open fractions, pump speeds, valve open fractions, valve settings, and, when needed, relaxed valve mode weights. Thus pump on/off actions in \texttt{Net1} and \texttt{Net3} become ramps over a finite window in the DAE rather than instantaneous jumps. The demand patterns are applied as in EPANET, that is, constant over each pattern step, and Theorem~\ref{thm:local_index_one} applies between two pattern changes. At a pattern change the junction balance~\eqref{eq:full_hydraulic_dae_junctions} itself changes, and the flows jump to satisfy it within one implicit Euler step, similar to EPANET. Unlike the pump and valve switches of Section~\ref{sec:smoothed_switching}, this jump is not smoothed, and the heads computed at that step depend on the step size. For control valves, the case studies only require dissipative valve behavior in the active suite: \texttt{BELL\_CL1} contains one TCV, which EPANET represents through a loss coefficient that depends on the valve setting rather than through a pressure or flow regulating setpoint. When a regulating valve of type FCV, PRV, or PSV is open but no longer actively regulating, the implementation replaces the ideal regulator by a quadratic loss equation for the open mode calibrated from the EPANET operating point,
\begin{equation}\label{eq:open_mode_valve_calibration}
\Delta h_v=r_v^{\mathrm{op}} q_v|q_v|,
\qquad
r_v^{\mathrm{op}}=\frac{|\Delta h_v^\star|}{|q_v^\star|^2+\varepsilon_q^2}.
\end{equation}
This is consistent with the way EPANET treats an inactive regulating valve, where the ideal regulator is dropped and a fixed minor loss representation is used~\cite[Sec.~3.1]{rossman2020epanet}. The constant $\varepsilon_q=10^{-5}$~m$^3$/s prevents division by zero at near stagnant operating points and has no effect outside that regime. 

The 24 hour errors in Tab.~\ref{tab:epanet_dae_error_suite} are larger than the one hour errors because tanks move, controls change, and the DAE uses the smoothed schedule~\eqref{eq:smoothed_compact_dae} rather than EPANET's instantaneous switching. This behavior is consistent with Theorem~\ref{thm:smoothed_switching}. The smoothing replaces a discontinuous control action with a differentiable one, so the reduced DAE vector field \eqref{eq:local_reduced_ode} is well defined inside the switching window and the implicit Euler step does not have to step across a jump in the residual. It does not, however, claim exact agreement with EPANET inside the window, because EPANET still uses an instantaneous switch there. Across the studied networks the $24$~h pressure errors split along network size and control logic. The \texttt{Anytown} and \texttt{MOD} simulations are at or below their $1$~h fixed control levels because their EPANET schedules trigger no pump switches over the horizon and the smoothed pump speeds match the patterns directly. The \texttt{Threenodes}, \texttt{Net1}, and \texttt{Net3} increase by one order of magnitude since each has tanks plus pumps controlled by tank water levels. The smoothed schedule replaces an instantaneous switch with a finite window ramp, and the largest pressure deviations therefore concentrate around those switching events. The $4.28$~m peak for \texttt{Net1} corresponds to a single pump on/off actuation that drives a tank level overshoot of a few meters during the smoothing window. The $1.05$~m peak for \texttt{Net3} corresponds to two pump cycles in the EPANET schedule. Fig.~\ref{fig:epanet_comparison}(b) and (c) show that the normalized total error $\|\be(t)\|_2/n_x$ stays at the $10^{-3}$ level outside switching events and hourly demand pattern changes for both networks, thus the largest junction head error is concentrated in narrow windows rather than reflecting a steady drift. The \texttt{BELL\_CL1} result mainly tests demand pattern changes and TCV loss settings, while \texttt{Net1} and \texttt{Net3} test pump switching. That being said, the results show acceptable recovery of flows and heads throughout the network for extended period simulations.

\begin{table*}[t]
\centering
\caption{Maximum errors over one hour simulations with fixed controls and 24 hour simulations with smoothed controls. The pressure error is the maximum over junction pressure head states $p^J$, and the flow error is the maximum over supported link flows.}\label{tab:epanet_dae_error_suite}
\renewcommand{\arraystretch}{1}
\setlength{\tabcolsep}{2.2pt}
\fontsize{7.0}{7.7}\selectfont
\begin{tabularx}{\textwidth}{>{\centering\arraybackslash\color{mainblue}\bfseries}m{0.19\textwidth} *{4}{>{\centering\arraybackslash}X}}
\blueRule
\rowcolor{aliceblue}
\textbf{\color{black}Network} &
\multicolumn{2}{c}{\textbf{\color{black}$1$ h fixed control}} &
\multicolumn{2}{c}{\textbf{\color{black}$24$ h smoothed control}} \\ \thinBlueRule
\rowcolor{aliceblue}
\textbf{\color{black}} &
\textbf{\color{black}\makecell{$p^J$ (m)}} &
\textbf{\color{black}\makecell{$q$ (m$^3$/s)}} &
\textbf{\color{black}\makecell{$p^J$ (m)}} &
\textbf{\color{black}\makecell{$q$ (m$^3$/s)}} \\ \midrule
\texttt{Threenodes} & $8.48{\times}10^{-3}$ & $3.38{\times}10^{-6}$ & $2.01{\times}10^{-1}$ & $1.30{\times}10^{-5}$ \\ \thinBlueRule
\texttt{Net1} & $2.25{\times}10^{-1}$ & $6.46{\times}10^{-4}$ & $4.28{\times}10^{0}$ & $1.78{\times}10^{-2}$ \\ \thinBlueRule
\texttt{Net3} & $1.60{\times}10^{-1}$ & $2.04{\times}10^{-3}$ & $1.05{\times}10^{0}$ & $8.03{\times}10^{-2}$ \\ \thinBlueRule
\texttt{Anytown} & $4.11{\times}10^{-2}$ & $2.63{\times}10^{-3}$ & $9.40{\times}10^{-2}$ & $1.48{\times}10^{-3}$ \\ \thinBlueRule
\texttt{MOD} & $5.33{\times}10^{-1}$ & $6.32{\times}10^{-4}$ & $1.57{\times}10^{-2}$ & $1.15{\times}10^{-5}$ \\ \thinBlueRule
\texttt{BELL\_CL1} & $8.00{\times}10^{-1}$ & $1.67{\times}10^{-2}$ & $8.31{\times}10^{-1}$ & $7.18{\times}10^{-3}$ \\
\blueRule
\end{tabularx}
\end{table*}

 Furthermore, linearizing the link equation~\eqref{eq:graph_form_dae_links} at the operating point gives the inertial time constant $\tau_e=(\gamma_e\kappa_e)^{-1}$ of link $e$. For most links of \texttt{Net3} and \texttt{Anytown}, $\tau_e$ is much shorter than the EPANET hydraulic time step $\Delta t_E=300$~s, which explains why the two models agree away from the switching windows. That being said, the two models are expected to differ when a pump or valve changes faster than $\tau_e$ on an important link, for example a pump trip or a valve closure on a trunk main. To that end, Fig.~\ref{fig:epanet_comparison}(d) considers a trip of the Lake pump of \texttt{Net3} (pump~\texttt{10}, $215.7$~L/s before the trip) at constant demand, with a smoothing width of $10$~s and all other controls fixed. The DAE is integrated with a $2$~s step for $15$~min after the trip, while EPANET keeps its $300$~s step.

 \begin{figure*}[!tp]
	\centering
	\includegraphics[width=\textwidth]{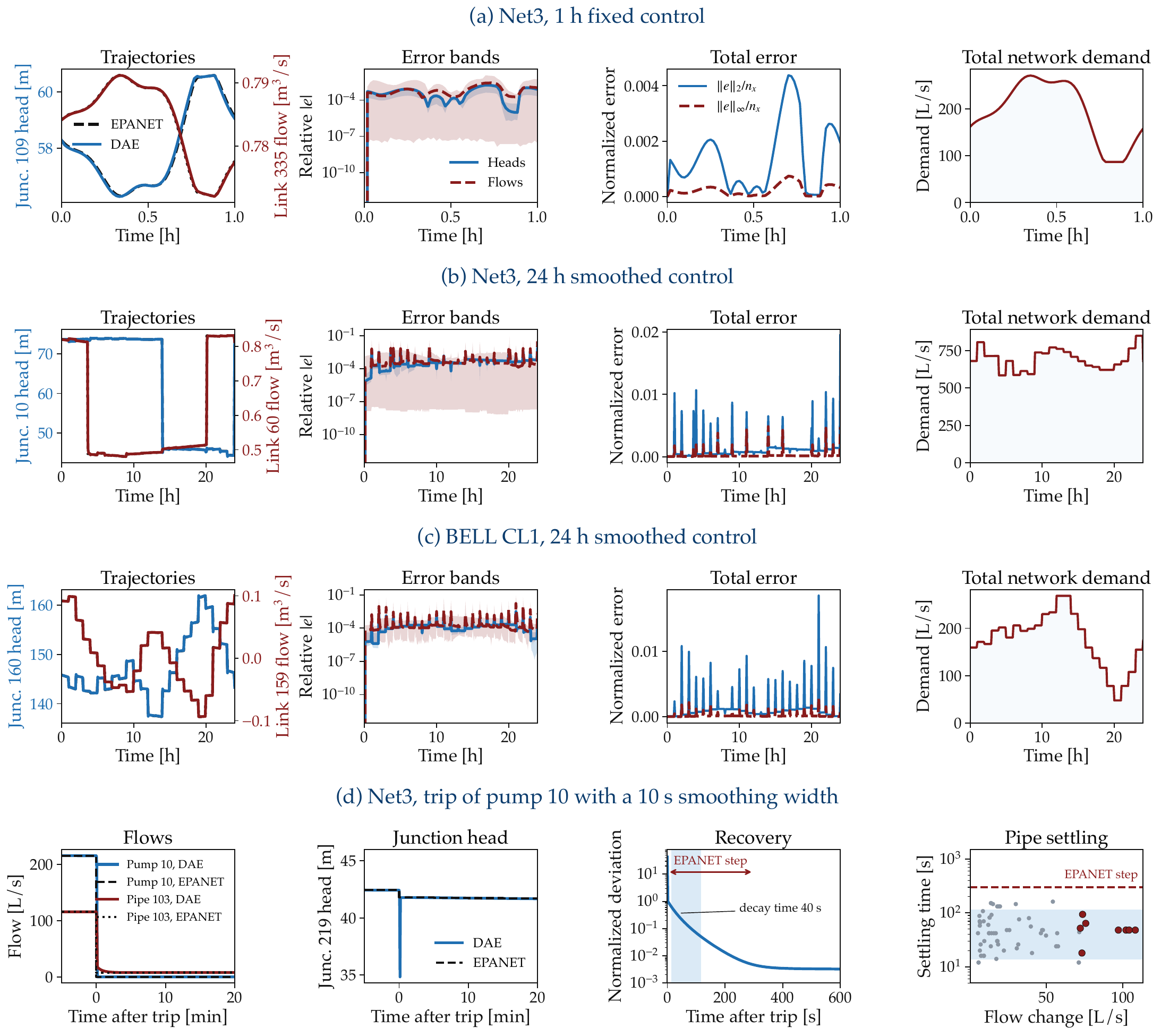}
	\caption{Comparison of the DAE water network representation and EPANET. Rows (a) to (c) depict the one hour fixed control run for \texttt{Net3} (Theorem~\ref{thm:local_index_one}) and the $24$ hour smoothed control runs for \texttt{Net3} and \texttt{BELL\_CL1} (Theorem~\ref{thm:smoothed_switching}). From left to right, each row shows a junction head and a link flow (DAE solid and EPANET dashed), the relative errors of heads (blue) and flows (red) having a median as a line and the $10$ to $90$ percentile range as a band, the total error norms $\|e(t)\|_2/n_x$ and $\|e(t)\|_\infty/n_x$, and the total demand. Legends are shown in row (a) only. In rows (b) and (c), the spikes in $\|e(t)\|_2/n_x$ occur at the hourly demand pattern changes and, in row (b), also at the two pump switching events. Row (d) shows the trip of pump \texttt{10} on \texttt{Net3} over a $10$~s smoothing width; the DAE integrated with a $2$~s step (solid) and EPANET having a $300$~s step (dashed and dotted). From left to right, it shows the flows of the tripped pump and of the most affected pipe, the head at the junction with the largest head change, the deviation of the DAE from the quasi-steady state normalized to one at the trip, and the settling time of each pipe with a flow change above $5$~L/s. In the third panel, the band is the range of inertial time constants of the linearized DAE, the line marks the fitted decay time. In the fourth panel, the eight most affected pipes are in red, the other pipes are in gray, and the EPANET step value is dashed.}
	\label{fig:epanet_comparison}
\end{figure*}

By observing Fig.~\ref{fig:epanet_comparison}(d), we highlight the following. \textit{(i)} EPANET moves to the post trip equilibrium within one step, thus its flow error at the trip sample equals the full pump flow and falls to $3.7\times10^{-4}$~m$^3$/s at the next sample. The DAE resolves the transient in between. \textit{(ii)} The junction heads drop at once, while the flows follow with their inertia. This is expected since the heads are given by the hidden algebraic constraint~\eqref{eq:hidden_constraint} that contains the pump head gain. This results in a head change as soon as the pump speed drops. The head at junction \texttt{219} dips from $42.4$~m to $34.8$~m eight seconds after the trip before settling at $45.4$~m, a transient pressure drop of $7.6$~m below the pre trip head that a static solve does not report, while the flow in pipe \texttt{103} falls from $115.9$ to $7.9$~L/s. \textit{(iii)} The deviation of the flows and junction heads from the quasi-steady state decays with a time constant of $40$~s. This value lies within the range of the time constants of the $25$ inertial modes of the linearized DAE ($13.6$ to $116$~s), and the deviation falls below five percent after $118$~s and below one percent before the EPANET step ends. \textit{(iv)} We define the settling time of a pipe flow as the time that its deviation from the quasi-steady flow falls to $1/e$ of its value at the trip. For the $69$ pipes with a flow change of more than $5$~L/s, the median settling time is $42$~s and the longest is $162$~s, and $83\%$ of these pipes settle within this range. Notice that a pipe does not settle at its own time constant $\tau_e$. The pipes on the supply path of the tripped pump (\texttt{103}, \texttt{109}, \texttt{105}, and \texttt{111}) settle together in $48$~s; their $\tau_e$ ranges from $30$ to $257$~s. This is due to the fact that the network graph couples the links and thus the recovery follows the modes of the linearized DAE dynamics. As such, the modes of $(\bA_h,\bE_h)$ in Section~\ref{sec:graph_theory} set the transient, and on a network with modes slower than $\Delta t_E$ EPANET would remain inaccurate for more than one step. Finally, the trip reduces the stability margin $\alpha_s$ of the storage modes from $3.47\times10^{-6}$ to $2.45\times10^{-6}$~s$^{-1}$ and the finite-horizon pump influence $G_H$ of Section~\ref{sec:case_operating_point_margins} from $9.82$ to $9.40$.

We note that the pressure drop and its recovery after the trip are not captured by a quasi-steady state model, whereas the linearized DAE predicts their duration from the operating point alone. A pipe burst or a manipulated pump setting results in a similar transient. As such, the DAE can provide the expected pressure response to such events, and comparing it with measured pressures can therefore help detect such events.

\begin{figure*}[t]
	\centering
	\includegraphics[width=0.8\textwidth]{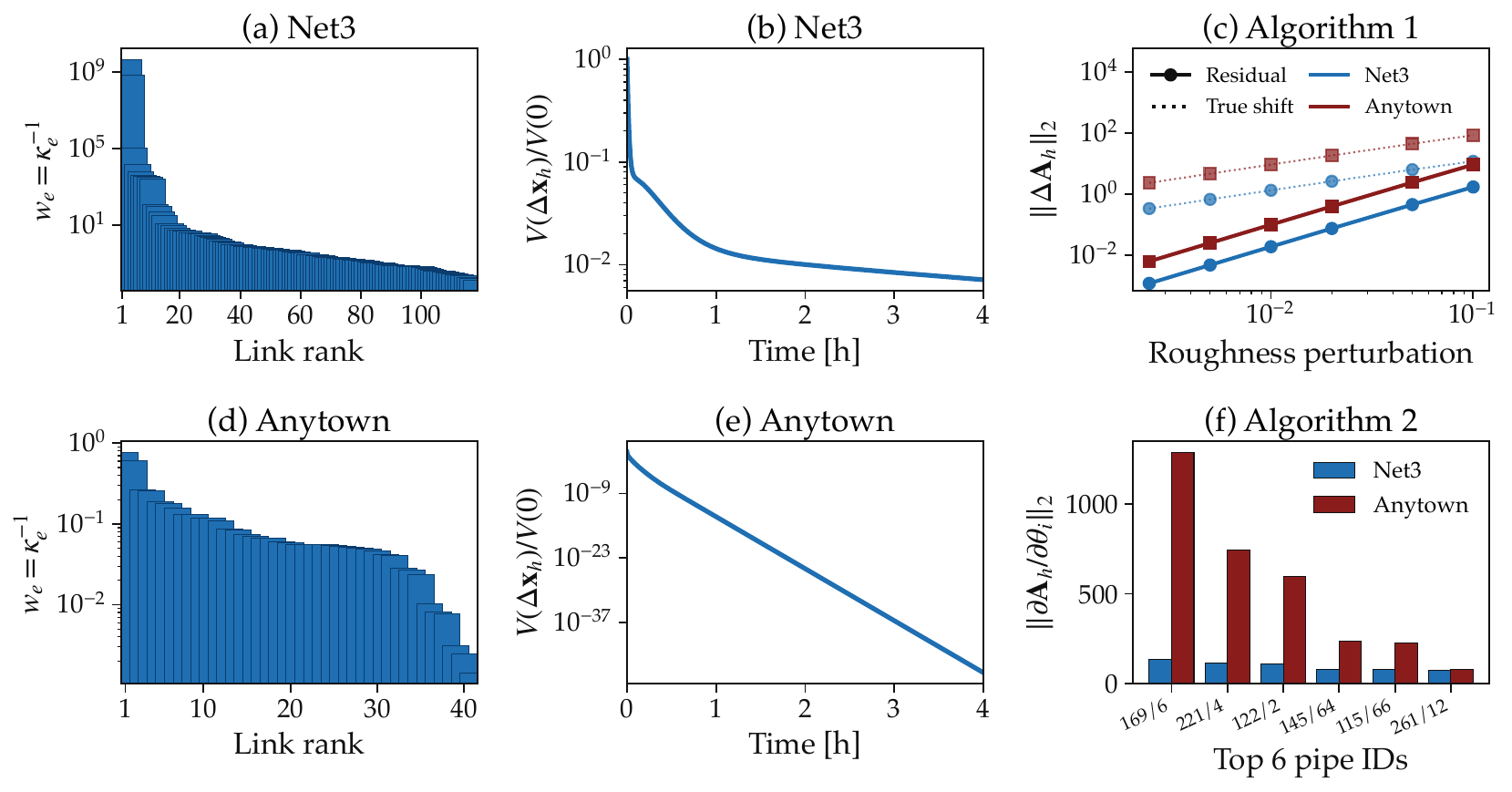}
	\caption{Stability and roughness sensitivity for \texttt{Net3} (top row) and \texttt{Anytown} (bottom row) at their fixed control operating points. Panels (a) and (d) show the ranked conductance weights $w_e=\kappa_e^{-1}$, and panels (b) and (e) show that the normalized incremental hydraulic energy decays as Theorem~\ref{thm:graph_linear_dae} predicts. Panel (c) compares the error of the first-order approximation of $\bA_h$ in Algorithm~\ref{proc:linearization_screen} (solid) with the true matrix shift (dotted). Panel (f) shows the six pipes with the largest sensitivity in Algorithm~\ref{proc:margin_screen}, with \texttt{Net3} and \texttt{Anytown} identified pipes.}
	\label{fig:stability_and_procedures}
\end{figure*}

\subsection{Local Stability Through the Weighted Laplacian}\label{sec:case_theorem5}

Theorem~\ref{thm:graph_linear_dae} connects the local linearized DAE to a weighted graph Laplacian and proves asymptotic stability through an incremental hydraulic energy Lyapunov function. To test this numerically, we construct the linearized DAE matrices $(\bE_h,\bA_h)$ and the weighted Laplacian $\bL_w=\bN\bW\bN^\top$ with $\bW=\bK^{-1}$ at the \texttt{Net3} fixed control operating point, and we repeat the linearization for \texttt{Anytown}. This operating point is the equilibrium with fixed controls at the nominal demand of each network ($192.6$~L/s in total for \texttt{Net3} and $403.8$~L/s for \texttt{Anytown}), at which the tank levels do not change. We note that the pump trip of Section~\ref{sec:case_theorem3} and the demand of Section~\ref{sec:case_operating_point_margins} linearize at operating points with given tank levels, that allow the tanks to fill or drain. The slow storage modes, that is $\alpha_s$, differ between the two cases. This explains why the stability margins reported in those sections are smaller. All finite eigenvalues of $(\bA_h,\bE_h)$ lie in the open left half plane for both networks, with stability margin $\alpha_s=2.06\times10^{-5}$ for \texttt{Net3} and $\alpha_s=3.67\times10^{-3}$ for \texttt{Anytown}. The spectrum of $\bL_w$ reflects active graph connectivity. In these simulations, $\bL_w$ has one near zero eigenvalue for \texttt{Anytown} and none for \texttt{Net3} at the $10^{-8}$ threshold. The finite eigenvalue and energy dissipation checks satisfy the theorem's stability conclusion.

To test the time domain prediction, a perturbation is applied to the equilibrium and the linearized DAE is integrated forward with an implicit Euler step. The incremental hydraulic energy $V(\Delta\bx_h)=\tfrac{1}{2}\Delta\bx_h^\top\bE_h\Delta\bx_h$ decays monotonically on both networks, and the dissipation rate $\dot V=-\Delta\bq^\top\bK\Delta\bq$ remains nonpositive throughout the four hour horizon. This corroborates the Lyapunov identity derived in Theorem~\ref{thm:graph_linear_dae}.

Fig.~\ref{fig:stability_and_procedures}(a), (b), (d), and (e) summarize the numerical assessment. Panels (b) and (e) show that the incremental hydraulic Lyapunov function behaves exactly as the theorem predicts on both networks, decaying monotonically over the four hour horizon. Panels (a) and (d) show the ranked conductance weights $w_e=\kappa_e^{-1}$ at the operating point. The two networks show very different operating regimes. For \texttt{Anytown} the weights are concentrated within a narrow range, whereas for \texttt{Net3} they spread over many orders of magnitude (the underlying range of $\kappa_e$ is reported quantitatively in Tab.~\ref{tab:procedure_theorem_case_study}). This spread is important since the weights $w_e$ form the diagonal matrix $\bW$ in $\bL_w=\bN\bW\bN^\top$. That is, the smallest weights set the slowest dissipation modes and the largest weights set the fastest. Each weight in the Laplacian is the incremental hydraulic conductance of an active link at the operating point. As such, the spread of the weights across the network sets the spread of dissipation rates in $\dot V=-\Delta\bq^\top\bK\Delta\bq$. We note that a single Lyapunov function covers all links of the network despite this spread.
From a network operation perspective, for \texttt{Net3}, notice that a few links dominate the slow hydraulic recovery modes, thus pressure monitoring or calibration around those links is more valuable than treating all pipes as equally informative. For \texttt{Anytown}, the narrower conductance spread indicates a more uniform local dissipation pattern, which is useful for design because local parameter changes are less likely to create isolated slow modes. As such, Theorem~\ref{thm:graph_linear_dae} identifies at a given operating point if all links of the network have similar conductances or if the recovery of the network is governed by a few weakly conducting links, which act as bottlenecks.

\subsection{Corroborating Error Bounds for Parameter Perturbations}\label{sec:case_theorem6}

This section focuses on Theorem~\ref{thm:linearization_bounds} and asks how the graph-theoretic results would be used by an engineer before trusting a linearized DAE. The graph DAE matrix in Theorem~\ref{thm:graph_linear_dae} is constructed at the nominal equilibrium of each network, Algorithm~\ref{proc:linearization_screen} is applied to a global roughness perturbation, and Algorithm~\ref{proc:margin_screen} is applied to each pipe separately. We perturb pipe roughness because in looped distribution networks the single pipe hydraulic resistances are not directly observable from the kind of pressure/flow data that is typically available, and roughness is the parameter that classical calibration targets~\cite{berardi2021calibration}. Other sources of parametric uncertainty can be similarly tested. 

For Algorithm~\ref{proc:linearization_screen}, all Hazen-Williams roughness coefficients are multiplied by $1+\delta$. The exact perturbed matrix $\bA_h(1+\delta)$ is recomputed after solving the fixed control DAE equilibrium, and it is compared with the first-order prediction
\[
\bA_h(1+\delta)\approx \bA_h(1)+
\left.\frac{\partial \bA_h}{\partial \theta}\right|_{\theta=1}\delta .
\]
For Algorithm~\ref{proc:margin_screen}, let $\theta_i$ denote the multiplicative roughness factor applied to pipe $i$, with $\theta_i=1$ at the nominal operating point. Each pipe roughness is perturbed one at a time and the local sensitivity $\|\partial \bA_h/\partial\theta_i\|_2$ is recorded. The DAE stability margin is evaluated from the finite generalized eigenvalues of $(\bA_h,\bE_h)$, and the Popov-Belevitch-Hautus (PBH) margin over the finite modes is computed from the smallest singular value of $[\lambda\bE_h-\bA_h,\ \bB_u]$ over those eigenvalues.

\begin{table*}[t]
\centering
\caption{Stability, linearization, and controllability quantities of \texttt{Net3} and \texttt{Anytown} at their fixed control operating points. The equilibrium residuals (infinity norm) are $1.17\times10^{-13}$ for \texttt{Net3} and $2.87\times10^{-10}$ for \texttt{Anytown}.}
\label{tab:procedure_theorem_case_study}
\renewcommand{\arraystretch}{1.10}
\setlength{\tabcolsep}{4pt}
\fontsize{6.8}{7.5}\selectfont
\begin{tabularx}{\textwidth}{>{\centering\arraybackslash\color{mainblue}\bfseries}m{0.18\textwidth} >{\centering\arraybackslash}m{0.14\textwidth} >{\centering\arraybackslash}m{0.14\textwidth} >{\centering\arraybackslash}X}
\blueRule
\rowcolor{aliceblue}
\textbf{\color{black}Quantity} & \textbf{\color{black}\texttt{Net3}} & \textbf{\color{black}\texttt{Anytown}} & \textbf{\color{black}Meaning} \\ \midrule
$\alpha_s$ (s$^{-1}$) & $2.06{\times}10^{-5}$ & $3.67{\times}10^{-3}$ & Slowest recovery rate, and $1/\alpha_s$ is the longest recovery time \\ \thinBlueRule
$\kappa(\bV)$ & $18.0$ & $35.9$ & Eigenvector conditioning of the finite modes in~\eqref{eq:stability_margin_perturbation} \\ \thinBlueRule
$[\min\;\kappa_e,\max\;\kappa_e]$ (s/m$^2$) & \makecell[c]{$2.17{\times}10^{-10}$\\ $73.3$} & \makecell[c]{$1.29$\\ $693$} & Range of incremental link slopes in Theorem~\ref{thm:graph_linear_dae} \\ \thinBlueRule
Largest relative residual & $0.145$ & $0.109$ & Algorithm~\ref{proc:linearization_screen} error at a $10\%$ global roughness perturbation \\ \thinBlueRule
Top pipes by $g_{A,i}$ & \texttt{169}, \texttt{221}, \texttt{122} & \texttt{6}, \texttt{4}, \texttt{2} & Largest effect on recovery, each pipe changed separately \\ \thinBlueRule
PBH margin & $2.03{\times}10^{-4}$ & $7.50{\times}10^{-5}$ & Smallest DAE controllability singular value over the finite modes \\ \thinBlueRule
$\sigma_c$ & $1.88{\times}10^{-15}$ & $1.51{\times}10^{-20}$ & Smallest singular value of $\boldsymbol{\mathcal{K}}_N$ over $H=2$~h (largest $2.70$ and $1.00$) \\ \thinBlueRule
Top pipes by $g_{K,i}$ & \texttt{329}, \texttt{111}, \texttt{105} & \texttt{2}, \texttt{4}, \texttt{6} & Largest effect on pump reach, each pipe changed separately \\
\blueRule
\end{tabularx}
\end{table*}

By observing Fig.~\ref{fig:stability_and_procedures}(c) and (f) and Tab.~\ref{tab:procedure_theorem_case_study}, we highlight the following. \textit{(i)} The relative residual between the true perturbed matrix and its first-order approximation grows from $3.6\times10^{-3}$ at $\delta=0.25\%$ to $1.45\times10^{-1}$ for \texttt{Net3}, and from $2.7\times10^{-3}$ to $1.09\times10^{-1}$ over the same range for \texttt{Anytown}. Both curves stay well below the true matrix shift $\|\Delta\bA_h\|_2$, which is consistent with the bound in~\eqref{eq:linearization_matrix_bound}. The linearization is therefore accurate for modest roughness uncertainty, and around ten percent network wide shift is the point where relinearizing at the perturbed equilibrium becomes safer than reusing the nominal matrix. \textit{(ii)} The pipes with the largest matrix sensitivity are not necessarily the pipes that have the largest absolute flow. They are links whose roughness changes produce a large change in the incremental headloss slope at the operating point. For \texttt{Anytown} the top three are the pump discharge trunks (pipes \texttt{6}, \texttt{4}, \texttt{2}), while for \texttt{Net3} they are interior pipes carrying modest flows. That being said, these are candidates for calibration attention, where a pipe can dominate, in terms of hydraulics, the local stability and controllability matrix even when a flow map of the network provides no visual cue of its importance. Here, what matters for $\bA_h$ is the operating point slope $\kappa_e = \partial \eta_e / \partial q_e$ rather than the flow magnitude itself.

\begin{figure*}[t]
	\centering
	\includegraphics[width=0.8\textwidth]{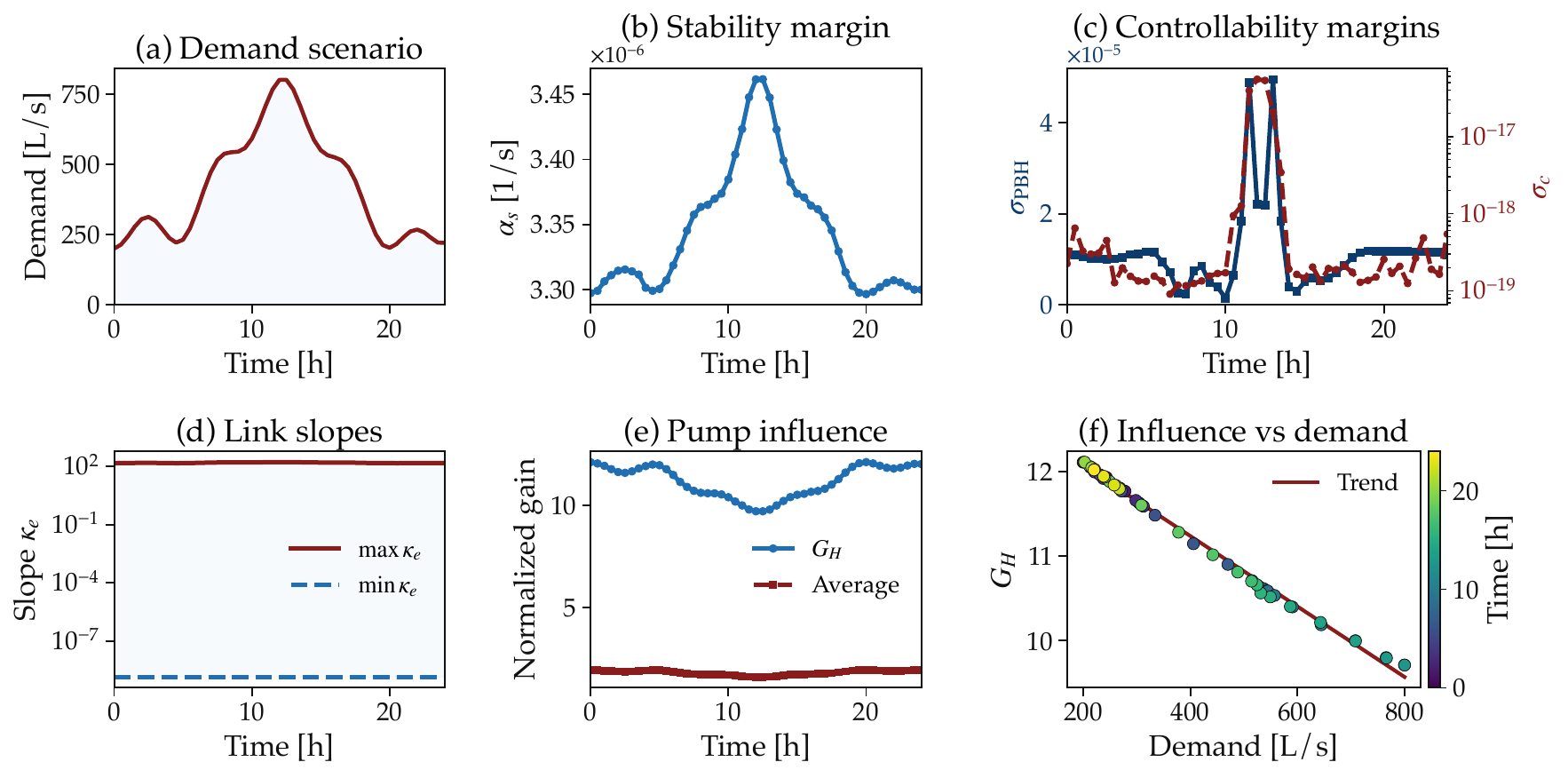}
	\caption{Margins of \texttt{Net3} under an aggressive demand scenario from $200$ to $800$~L/s, with the EPANET spatial demand split kept. At each demand level the DAE is linearized at its operating point. Panel (a) shows the total demand, panel (b) the stability margin $\alpha_s$, panel (c) the PBH margin with the controllability margin $\sigma_c$ on the right axis, and panel (d) the range of the slopes $\kappa_e$. Panel (e) shows the pump influence $G_H$ over a two hour horizon and the average gain over all pump inputs and time steps (dashed), and panel (f) shows $G_H$ against total demand.}
	\label{fig:net3_hourly_study}
\end{figure*}

We next test the controllability part of Theorem~\ref{thm:property_robustness} at the same operating points. The margin $\sigma_c$ in~\eqref{eq:controllability_margin_def} is evaluated for the implicit Euler discretization of the reduced model over $H=2$~h with a $60$~s step. Given that the flows must satisfy the linearized junction mass balance, then $\sigma_c$ computed only over the deviations $\Delta\bz$ must satisfy this balance. By Theorem~\ref{thm:local_index_one}, there are $n_q+n_A-n_J$ such deviations, that is, $29$ for \texttt{Net3} and $22$ for \texttt{Anytown}. By observing Tab.~\ref{tab:procedure_theorem_case_study}, we highlight the following. \textit{(i)} Both networks are controllable over the horizon, as the PBH margins confirm. However, $\sigma_c$ is only $1.9\times10^{-15}$ for \texttt{Net3} and $1.5\times10^{-20}$ for \texttt{Anytown}, compared with largest singular values of $2.7$ and $1.0$, and $14$ of the $29$ consistent directions of \texttt{Net3} and $14$ of the $22$ directions of \texttt{Anytown} have singular values below $10^{-6}$. The existing pumps can change only a few combinations of pressures, flows, and tank levels, while about half of these combinations barely respond to the pumps within two hours. \textit{(ii)} The eigenvector conditioning is $\kappa(\bV)=18.0$ for \texttt{Net3} and $35.9$ for \texttt{Anytown}. As such, the constant in~\eqref{eq:stability_margin_perturbation} is moderate. Nevertheless, the guaranteed remaining margins of Algorithm~\ref{proc:margin_screen} are not informative at the $5\%$ roughness change of Appendix~\ref{app:numerical_settings}. The term $\kappa(\bV)g_{A,i}|\delta\theta_i|$ exceeds $\alpha_s$ by several orders of magnitude for every pipe, since $g_{A,i}$ measures the change of the full matrix $\bA_h$, which is dominated by the link inertia terms, whereas $\alpha_s$ is set by the slow storage modes. The margins recomputed after each single pipe change move by less than two percent ($\alpha_s$ from $2.057\times10^{-5}$ to at most $2.088\times10^{-5}$ on \texttt{Net3}), and $\sigma_c$ keeps the same order of magnitude. As such, Algorithm~\ref{proc:margin_screen} is used as a ranking. \textit{(iii)} Ranked by $g_{K,i}=\|\partial\boldsymbol{\mathcal{K}}_N/\partial\theta_i\|_2$, the top pipes are \texttt{329}, \texttt{111}, and \texttt{105} on \texttt{Net3} and \texttt{2}, \texttt{4}, and \texttt{6} on \texttt{Anytown}, which lie on the supply paths of the pumps. Two of the \texttt{Net3} pipes in this ranking (\texttt{105} and \texttt{111}) are also among the pipes most affected by the pump trip of Section~\ref{sec:case_theorem3}. This shows that, on \texttt{Net3}, the pipes that govern stability, ranked by $g_{A,i}$, and the pipes that govern controllability, ranked by $g_{K,i}$, are different, and the latter are the pipes a utility would protect to preserve the reach of its pumps. 

\begin{figure*}[t]
	\centering
	\includegraphics[width=\textwidth]{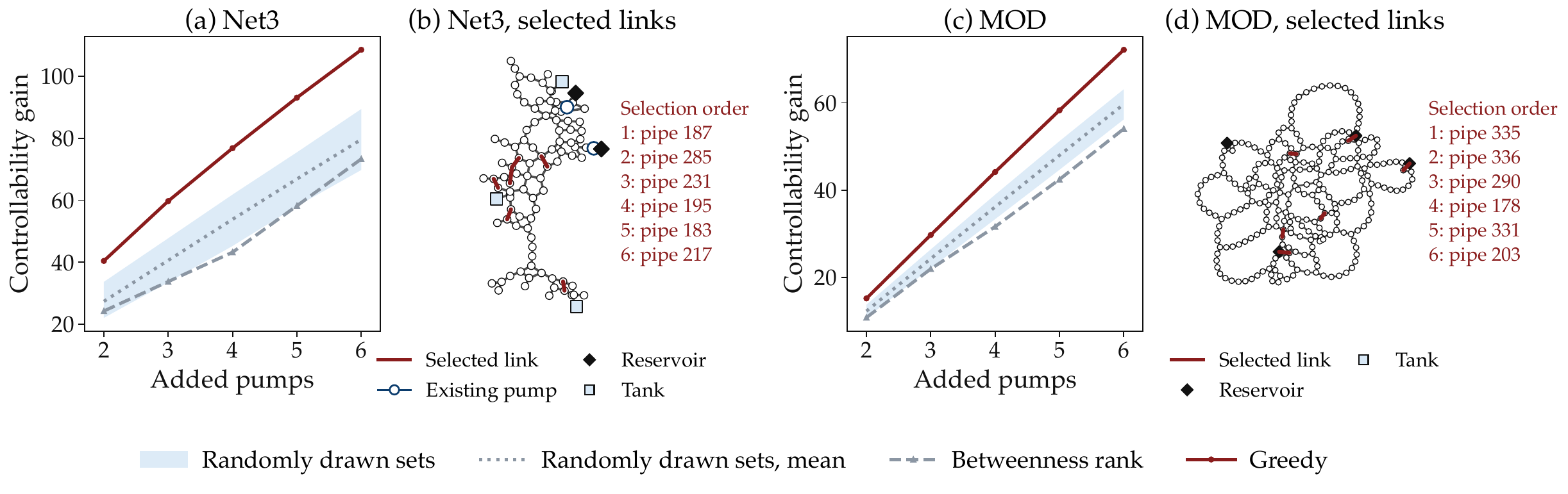}
	\caption{Pump allocation on \texttt{Net3} (a, b) and \texttt{MOD} (c, d). Panels (a) and (c) show the controllability gain of the greedy algorithm, the betweenness rank, and randomly drawn sets (mean and spread). Panels (b) and (d) show the six links selected by the greedy algorithm.}
	\label{fig:relevance_case_studies}
\end{figure*}

\subsection{Stability and Controllability Margins Under Demand Variation}\label{sec:case_operating_point_margins}

The previous section keeps the operating point fixed. We now ask how the same DAE stability and controllability quantities change when demand moves the network through different local operating points. For \texttt{Net3}, we impose an intentionally aggressive total demand profile over a 24 hour period, ranging from $200$ to $800$~L/s and containing both slow daily variation and smaller oscillations. This profile is \textit{not} meant to represent a normal forecast; it is an aggressive demand scenario that makes the operating point move enough to reveal whether the margins are sensitive. At each sample, the EPANET junction demand vector is rescaled by one scalar, so the spatial demand split remains inherited from \texttt{Net3}.  The graph DAE matrices $(\bE_h,\bA_h,\bB_u)$ are reconstructed, and the DAE stability margin is recorded as
$
\alpha_s(t)=-\max_{\lambda\in\Lambda_f(t)}\Re(\lambda),
$
where $\Lambda_f$ is the set of finite generalized eigenvalues of $(\bA_h,\bE_h)$ and $\Re(\lambda)$ denotes the real part of $\lambda$. The PBH controllability margin over the finite modes is
\[
\sigma_{\mathrm{PBH}}(t)=
\min_{\lambda\in\Lambda_f}
\sigma_{\min}\!\left([\lambda\bE_h-\bA_h,\ \bB_u]\right).
\]
These are the DAE analogues of the stability and controllability quantities used in Theorem~\ref{thm:property_robustness}, evaluated across operating points rather than across parameter perturbations. Since the pump schedule itself remains fixed, $\sigma_{\mathrm{PBH}}$ is best interpreted as a local rank margin for infinitesimal pump speed perturbations. To add a more operational strength measure, let $t_i$ denote one demand sample and define $\bE_i=\bE_h(t_i)$, $\bA_i=\bA_h(t_i)$, and $\bB_i=\bB_u(t_i)$. We hold this linearized model fixed only over a local horizon $H=2$~h with step $\tau$, compute a finite-horizon reachability metric, and then repeat the same calculation at the next sample of the 24 hour computation
\[
(\bE_i-\tau\bA_i)\Delta\bx_{k+1}
=\bE_i\Delta\bx_k+\tau\bB_i\Delta\bu_k.
\]
The corresponding reachability matrix is given by
\[
\bR_{H,i}=\left[\boldsymbol{\Gamma}_i, \; \boldsymbol{\Phi}_i\boldsymbol{\Gamma}_i, \; \ldots,\;\boldsymbol{\Phi}_i^{N_H-1}\boldsymbol{\Gamma}_i\right],
\]
where $\boldsymbol{\Phi}_i=(\bE_i-\tau\bA_i)^{-1}\bE_i$, $\boldsymbol{\Gamma}_i=(\bE_i-\tau\bA_i)^{-1}\tau\bB_i$, and $N_H=H/\tau$. The finite-horizon reachability gain, which we refer to as the pump influence, is
$
G_{H,i}=\sigma_{\max}\!\left(\bS_x^{-1}\bR_{H,i}\right),
$
where $\bS_x$ is a diagonal scaling matrix that scales flow states by $1$~m$^3$/s and head states by $10$~m.

Fig.~\ref{fig:net3_hourly_study} shows that aggressive demand changes \textit{do} affect the linearized model, however they \textit{do not} produce a stability or pump influence collapse. Over this scenario, the stability margin remains positive and changes only from $3.30\times10^{-6}$ to $3.46\times10^{-6}$ (note the strong positive correlation between the demand pattern and the change it produces in the stability margin). Unlike the stability margin, the PBH margin changes abruptly between samples, ranging from $1.42\times10^{-6}$ to $4.95\times10^{-5}$, but it remains positive throughout. The controllability margin $\sigma_c$ of Theorem~\ref{thm:property_robustness}, shown on the right axis of Fig.~\ref{fig:net3_hourly_study}(c), stays between $9\times10^{-20}$ and $5.5\times10^{-17}$ and increases with demand (correlation $0.57$). This is the opposite trend to $G_H$ discussed below. Higher demand weakens the directions that the pumps reach strongly and slightly strengthens the weakly reached ones, since larger head losses couple the network more evenly. The largest incremental slope changes from $1.51\times10^2$ to $1.71\times10^2$, while the smallest stays near $1.41\times10^{-9}$. The finite-horizon reachability gain gives the more physical picture. The value $G_{H,i}$ varies from $9.71$ to $12.1$, and the average singular value varies from $1.62$ to $1.96$. Furthermore, Fig.~\ref{fig:net3_hourly_study}(f) shows a monotone decrease of $G_{H,i}$ as total demand increases. For \texttt{Net3} aggressive demand scenario, higher demand therefore weakens pump influence in the sense that the same pump speed input energy produces a smaller normalized hydraulic state response. This is consistent with high flow operating points being more dissipative through larger incremental headloss effects. At the same time, the gain change is modest and the PBH margin stays positive, so the result does not suggest an imminent loss of controllability.

\subsection{Pump Allocation in WDNs}\label{sec:case_actuators}
Section~\ref{sec:case_theorem6} shows that about half of the combinations of pressures, flows, and tank levels of \texttt{Net3} are barely affected by existing pumps. To that end, we apply the greedy algorithm of Section~\ref{sec:linear_dae_properties} to find where an added pump increases that response the most. We compute the controllability gain, represented by the increase of $\log\det(\bW_N(\mathcal{S})+\epsilon\bI)$ with $\epsilon=10^{-6}$ over the existing pumps on \texttt{Net3}. Since \texttt{MOD} has no pump or valve, all selections start from pipe \texttt{335}, the best single location found by the greedy algorithm, and the gain is measured from it. Here $\epsilon$ corresponds to $1$~L/s of flow or $1$~cm of head. The greedy pump allocation is compared with the static graph-based \textit{betweenness rank}, that orders the links by edge betweenness centrality. It is defined as the fraction of shortest paths between node pairs that pass through the link~\cite{yazdani2011complex}. This rank can be computed from the network layout alone, without a hydraulic model. For reference, we also draw $500$ sets of each size uniformly at random and show their mean and their spread.

Fig.~\ref{fig:relevance_case_studies} shows that the greedy algorithm gives the largest controllability gain for every number of added pumps, above both the betweenness rank and the randomly drawn sets. These results show that the controllability of the DAE model indicates where an added pump extends the reach of the existing pumps, which the network layout alone does not. Betweenness favors links that lie on many shortest paths, whereas the greedy algorithm selects the links that move the combinations of pressures, flows, and tank levels that the existing pumps barely move. On \texttt{Net3}, the selected links carry modest flows, between $4$ and $93$~L/s, and are spread over the network away from the existing pumps, which indicates that the added pumps reach parts of the network that the existing pumps do not. On \texttt{MOD}, the greedy algorithm first selects the trunks that carry the reservoir outflows (pipes \texttt{335}, \texttt{336}, and \texttt{290}, that carry $222$, $56$, and $161$~L/s) before moving into the loops. As such, the proposed DAE model and its controllability matrix can inform the allocation of new pumps based on both the graph and the dynamics of the network, rather than on the network layout alone.

\subsection{Answering the Posed Research Questions}\label{sec:implications}
Having presented the above results, we now answer the posed research questions and state what each answer means for the operation of a water network. \textit{(i)} The one hour and 24 hour comparisons in Sections~\ref{sec:case_theorem1} and~\ref{sec:case_theorem3} show that the DAE reproduces the EPANET heads and flows away from the switching windows, with the largest deviations at the pump switching events. The pump trip in Section~\ref{sec:case_theorem3} shows the practicality of the DAE model. When the recovery after a pump or valve change is shorter than the EPANET hydraulic step, the DAE resolves the transient. Here, the DAE shows a $7.6$~m pressure dip and a $40$~s recovery after the trip. As such, the DAE can provide the expected pressure response to sudden events such as a pump trip, a pipe burst, or a manipulated pump setting. This can offer a detection method that quasi-steady state simulations do not offer; this answers $\mathrm{Q}1$. \textit{(ii)} The conductance weights in Section~\ref{sec:case_theorem5} identify the group of links that govern the slow recovery of pressures after a disturbance. This indicates the links where pressure monitoring is most informative; this answers $\mathrm{Q}2$. \textit{(iii)} The linearization study in Section~\ref{sec:case_theorem6} indicates that the linearized model remains accurate up to a network wide roughness change of roughly ten percent, beyond which it should be linearized again at the new operating point. It also ranks the pipes that most affect the local dynamics, that are not necessarily the pipes carrying the largest flows; this answers $\mathrm{Q}3$. \textit{(iv)} The demand scenario in Section~\ref{sec:case_operating_point_margins} shows that aggressive demand variation changes the stability margin mildly and reduces pump influence as demand grows. Thus, pump schedules designed at low demand should be checked at peak demand; this answers $\mathrm{Q}4$. \textit{(v)} The pump allocation study in Section~\ref{sec:case_actuators} shows that the controllability of the DAE model identifies where an added pump most increases the reach of the existing pumps, where topological methods that do not consider the hydraulic model cannot identify; this answers $\mathrm{Q}5$ and concludes this section.

\section{Conclusion, Paper Limitations, and Recommendations for Future Work}\label{sec:future}

This paper develops a control engineering perspective on the analysis of WDN hydraulics. The stability and controllability analysis explains how the layout, the hydraulic parameters, and the pumps of a network shape its response to uncertainty and disturbances. That said, this paper connects \textit{(i)} the well-posedness of the hydraulics under fixed controls, \textit{(ii)} the smoothing of pump and valve switching, \textit{(iii)} the weighted Laplacian structure of the linearized model, and \textit{(iv)} stability and controllability margins that depend on the hydraulic parameters. The value is that these foundations are tailored to water system dynamics and assembled into one DAE analysis framework that can be useful for operators and researchers. The case studies show that the linearized DAE explains how the network responds to pump actions. To conclude, \textit{(i)} for stability, the incremental hydraulic energy decays at every tested operating point, the stability margin stays positive under aggressive demand changes, and the recovery after a pump trip follows the inertial time constants of the linearized DAE; and \textit{(ii)} for controllability, the existing pumps of \texttt{Net3} can regulate only part of the network within two hours, and the controllability of the DAE model indicates where an added pump extends their reach. We also note that the pipes that govern network stability differ from those that govern the reach of the pumps.

This work is not devoid of limitations.  First, the theory is local in the operating point and in the active pump and valve modes. Although the case studies show that the margins are not extremely sensitive to the tested operating point changes, extending the analysis to the full nonlinear DAE model~\eqref{eq:full_hydraulic_dae}-\eqref{eq:compact_dae_form} remains important. Second, the model does not include water hammer, pressure dependent leakage, or other important phenomena in water systems. The rigid water column level is adequate when control actions are slower than the inertial time constants $\tau_e$ of the affected links and faster than the hydraulic time step of a quasi-steady state simulator; outside that window either EPANET suffices or a compressible transient model is required. Third, the methods diagnose stability and controllability margins and rank candidate pump locations through the greedy algorithm in Section~\ref{sec:linear_dae_properties}, but they do not prescribe pipe upgrades or other network modifications that would enlarge those margins. Future work should pursue hybrid system theory, pressure dependent demand and leakage models, and network modification methods that improve stability and controllability.

\bibliographystyle{IEEEtran}
\bibliography{refs}

\appendices

\section{Paper Notation}\label{app:notation}
Vectors are denoted in bold lowercase (e.g., $\bx$), matrices in bold uppercase (e.g., $\bA$), and node or link sets by calligraphic letters (e.g., $\mathcal{V}$). Bold Greek symbols are reserved for vector or matrix quantities, such as the hydraulic parameter vector $\boldsymbol{\vartheta}$, diagonal inertance matrices $\boldsymbol{\Gamma}^{\alpha}$, and the unified link nonlinearity $\boldsymbol{\eta}$. Scalar Greek symbols, such as $\gamma_e$ and $\kappa_e$, are not bold. We use $\|\cdot\|_p$ for $\ell_p$ norms, $\|\cdot\|$ for the Euclidean norm, and $|\bx|$ for the componentwise absolute value of a vector. Inequalities such as $\bx\le \by$ are interpreted componentwise. The operator $\mathrm{sign}(\cdot)$ is also applied componentwise. The notation $\operatorname{diag}(\bx)$ forms a diagonal matrix from a vector, $\sigma_{\min}(\cdot)$ denotes the smallest singular value, $\lambda_i(\cdot)$ denotes an eigenvalue, and $\kappa(\cdot)$ denotes the spectral condition number. Overbars denote nominal operating quantities, for example $(\bar{\bx},\bar{\bu},\bar{\bd})$, while $\Delta(\cdot)$ denotes deviations from those nominal values. A dot over a vector denotes a time derivative, where $\dot{\bz}$ is the derivative of the differential state and $\Delta\dot{\bz}$ is the derivative of its deviation. Quantity $\{\bx_1,\ldots,\bx_r\}=
\begin{bmatrix}
\bx_1^\top & \cdots & \bx_r^\top
\end{bmatrix}^{\!\top}$
denotes the stacked collection of compatible vectors. This convention is used repeatedly in the hydraulic model to avoid unnecessary transpose notation, while explicit column vectors are written with square brackets, for example $[1\ 0\ -1]^\top$. The symbol $\odot$ denotes the Hadamard (componentwise) product. Finally, the notation $\bX \succ 0$ means that matrix $\bX$ is a symmetric positive definite matrix, i.e., it has all positive eigenvalues (equivalently, all positive leading principal minors).

\section{Three Node Network}\label{app:threenodes}
\begin{figure}[t]
	\centering
	\includegraphics[width=0.6\columnwidth]{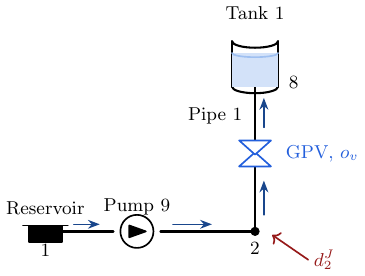}
	\caption{The \texttt{Threenodes} network with one reservoir, one pump, one junction, one pipe, and one tank. Its GPV is added for illustration and is kept fully open.}
	\label{fig:threenodes}
\end{figure}
This appendix develops the full DAE of Section~\ref{sec:methods} for the \texttt{Threenodes} benchmark shown in Fig.~\ref{fig:threenodes}. The goal is not to introduce a new model, but to make the general construction concrete by displaying the incidence blocks and the resulting equations explicitly. Here we use the physical EPANET labels of the benchmark, which has one pump ($9$), one pipe ($1$), one junction ($2$), one tank ($8$), and one reservoir ($1$), rather than introducing anonymous edge labels.
 Then the block incidence matrices are
\begin{equation}\label{eq:threenodes_incidence_blocks}
	\begin{aligned}
		\bC_J^M&=[-1],\qquad \bC_A^M=[0],\qquad \bC_R^M=[1],\\
		\bC_J^P&=[1],\qquad \bC_A^P=[-1],\qquad \bC_R^P=[0].
	\end{aligned}
\end{equation}
The GPV in Fig.~\ref{fig:threenodes} is added for illustration. It is placed in series with pipe~$1$ between junction~$2$ and tank~$8$ and is kept fully open. It connects the same two nodes in the same direction as pipe~$1$ and therefore has the same incidence entries
\begin{equation}\label{eq:threenodes_valve_incidence}
	\bC_J^W=[1],\qquad \bC_A^W=[-1],\qquad \bC_R^W=[0].
\end{equation}
With the control vector $\bu=\{s_9,o_v\}$ and the demand vector $\bd=\{d_2^J,d_8^A\}$, the DAE~\eqref{eq:full_hydraulic_dae} written for this network becomes
\begin{subequations}\label{eq:threenodes_dae}
	\begin{align}
		\dot q_9^{M}(t) &= \gamma_9^{M}\Big(p_1^{R}(t)-p_2^{J}(t)+\psi_9^{M}(q_9^{M}(t),s_9(t))\Big), \label{eq:threenodes_dae_pump}\\
		\dot q_1^{P}(t) &= \gamma_1^{P}\Big(p_2^{J}(t)-p_8^{A}(t)-\varphi_1^{P}(q_1^{P}(t))\nonumber\\ &\qquad-\varphi_v^{W}(q_1^{P}(t),o_v(t))\Big), \label{eq:threenodes_dae_pipe}\\
		\dot p_8^{A}(t) &= \frac{1}{a_8^{A}}\big(q_1^{P}(t)-d_8^{A}(t)\big), \label{eq:threenodes_dae_tank}\\
		\mathbf{0} &= -q_9^{M}(t)+q_1^{P}(t)+d_2^{J}(t), \label{eq:threenodes_dae_junction}\\
		\mathbf{0} &= p_1^{R}(t)-\bar p_1^{R}. \label{eq:threenodes_dae_reservoir}
	\end{align}
\end{subequations}
The corresponding link equations are
\begin{subequations}\label{eq:threenodes_component_laws}
	\begin{align}
		\varphi_1^{P}(q_1^{P}) &= r_1^{P}(q_1^{P})\,q_1^{P}|q_1^{P}|, \label{eq:threenodes_component_laws_pipe}\\
		\varphi_v^{W}(q_1^{P},o_v) &= r_v^{W}(o_v)\,q_1^{P}|q_1^{P}|, \label{eq:threenodes_component_laws_valve}\\
		\psi_9^{M}(q_9^{M},s_9) &= s_9^2\Big(h_9^0-r_9^{M}(q_9^{M}s_9^{-1})^{\nu_9}\Big). \label{eq:threenodes_component_laws_pump}
	\end{align}
\end{subequations}
Since the GPV is open, $r_v^{W}(o_v)=0$ and the valve loss term in~\eqref{eq:threenodes_component_laws_valve} vanishes. This term is kept in the equations above to show exactly where the valve head loss appears in the momentum equation of pipe~$1$ in~\eqref{eq:threenodes_dae_pipe}.

We next derive the graph linearized DAE in~\eqref{eq:linearized_graph_dae} for the same network, starting from the nonlinear equations in~\eqref{eq:threenodes_dae}. Let $\lambda_9^M=(\gamma_9^M)^{-1}$ and $\lambda_1^P=(\gamma_1^P)^{-1}$. At a fixed equilibrium $(\bar q_9^M,\bar q_1^P,\bar p_2^J,\bar p_8^A,\bar s_9,\bar o_v)$, define the local slopes
\begin{equation}\label{eq:threenodes_kappa}
\hspace{-0.55cm}\kappa_9=-\frac{\partial \psi_9^M}{\partial q_9^M}(\bar q_9^M,\bar s_9)>0,	\kappa_1=\frac{\partial}{\partial q_1^P}
		\big(\varphi_1^P+\varphi_v^W\big)(\bar q_1^P,\bar o_v)>0.
\end{equation}
The pump sign is negative because the pump head gain decreases with increasing flow on the fixed mode branch. With graph state
\[
\Delta\bx_h=
[\Delta q_9^M\ \Delta q_1^P\ \Delta p_2^J\ \Delta p_8^A]^\top,
\]
and with $\Delta\bu=\mathbf{0}$ and $\Delta\bd=\mathbf{0}$, linearizing~\eqref{eq:threenodes_dae} gives
\begin{equation}\label{eq:threenodes_linear_matrices}
	\displaystyle
	\bE_h\Delta\dot{\bx}_h=\bA_h\Delta\bx_h,
	\quad
	\bE_h=
	\begin{bmatrix}
		\lambda_9^M & 0 & 0 & 0\\
		0 & \lambda_1^P & 0 & 0\\
		0 & 0 & 0 & 0\\
		0 & 0 & 0 & a_8^A
	\end{bmatrix},
\end{equation} 
and $$
\bA_h=
\begin{bmatrix}
	-\kappa_9 & 0 & -1 & 0\\
	0 & -\kappa_1 & 1 & -1\\
	-1 & 1 & 0 & 0\\
	0 & 1 & 0 & 0 \end{bmatrix}.$$
The control input columns omitted from~\eqref{eq:threenodes_linear_matrices} are obtained by differentiating the two link equations with respect to $s_9$ and $o_v$, and the demand input columns enter the junction and tank rows. The corresponding head difference incidence blocks for the signed flows in this appendix are
\[
\widehat{\bN}_J=[-1\ \ 1],\qquad
\widehat{\bN}_A=[0\ \ -1],\qquad
\widehat{\bN}_R=[1\ \ 0],
\]
so that $\bK=\operatorname{diag}(\kappa_9,\kappa_1)$, $\bW=\bK^{-1}$, and $\widehat{\bL}_w=\widehat{\bN}\bW\widehat{\bN}^\top$. The static linearized link relation gives $\widehat{\bN}^\top\Delta\bp=\bK\Delta\bq$, hence $\Delta\boldsymbol{\iota}=\widehat{\bL}_w\Delta\bp$. The incremental energy is
\[
V=\tfrac{1}{2}\lambda_9^M(\Delta q_9^M)^2
+\tfrac{1}{2}\lambda_1^P(\Delta q_1^P)^2
+\tfrac{1}{2}a_8^A(\Delta p_8^A)^2.
\]
Using the algebraic constraint $-\Delta q_9^M+\Delta q_1^P=0$ in~\eqref{eq:threenodes_linear_matrices} yields
\[
\dot V=-\kappa_9(\Delta q_9^M)^2-\kappa_1(\Delta q_1^P)^2\le 0,
\]
which is the two link instance of Theorem~\ref{thm:graph_linear_dae}. Since the active component contains the reservoir node, the only invariant set with $\dot V=0$ is the origin on the set of states satisfying the algebraic constraints.

\section{Proof of Theorem~\ref{thm:local_index_one}}\label{app:proof_local_index_one}
\begin{proof}
We first show that the algebraic equations~\eqref{eq:compact_dae_form_algebraic} cannot be solved for the junction heads. The junction balance~\eqref{eq:full_hydraulic_dae_junctions} contains no head, and the reservoir condition~\eqref{eq:full_hydraulic_dae_reservoirs} contains the reservoir heads only. Hence $\partial_{\by}\bh$ has rank $n_R$ at every state, and the junction heads are not determined by the algebraic equations alone.

We next determine the junction heads from the derivative of the junction balance. Differentiating~\eqref{eq:full_hydraulic_dae_junctions} and inserting~\eqref{eq:full_hydraulic_dae_pipes} to~\eqref{eq:full_hydraulic_dae_valves} gives~\eqref{eq:hidden_constraint} with $\bL_J=\bC_J\boldsymbol{\Gamma}\bC_J^\top$. For a vector $\bm{v}\in\mathbb{R}^{n_J}$ we have $\bm{v}^\top\bL_J\bm{v}=\sum_{e}\gamma_e(\bC_J^\top\bm{v})_e^2$ with $\gamma_e>0$. Hence $\bL_J\bm{v}=\mathbf{0}$ if and only if $\bC_J^\top\bm{v}=\mathbf{0}$. This means that $\bm{v}$ takes the same value at the two ends of every open link between junctions and vanishes at every junction that is linked to a tank or a reservoir. Under Assumption~\ref{ass:regular_fixed_mode}\textit{(iii)} every junction is connected to such a node through open links, so $\bm{v}=\mathbf{0}$ and $\bL_J$ is nonsingular. Equation~\eqref{eq:hidden_constraint} therefore defines the junction heads as a continuously differentiable function of $\bz$, $\bu$, $\bar{\bp}^R$, and $\dot{\bd}^J$. Together with $\bp^R=\bar{\bp}^R$ this function is the map $\boldsymbol{\pi}_y$.

We now obtain the reduced vector field and the index. Substituting $\boldsymbol{\pi}_y$ into~\eqref{eq:full_hydraulic_dae_pipes} to~\eqref{eq:full_hydraulic_dae_tanks} gives~\eqref{eq:local_reduced_ode}. The map $\mathbf{F}$ is continuous in $t$ and locally Lipschitz in $\bz$, since the link equations are continuously differentiable and $\bd^J$ is continuously differentiable. The local existence and uniqueness theorem for ordinary differential equations~\cite[Thm.~2.2]{teschl2012ode} gives a unique solution $\bz(t)$ on $[t_0,t_0+\delta]$ for every $\bz(t_0)$. Along this solution $\bC_J\dot{\bq}=-\dot{\bd}^J$ holds by construction, so $\bC_J\bq(t)+\bd^J(t)$ is constant in time and equals zero when the initial condition is consistent. The heads follow from $\boldsymbol{\pi}_y$. Conversely, differentiating~\eqref{eq:full_hydraulic_dae_junctions} along any solution of the DAE shows that it satisfies~\eqref{eq:hidden_constraint}, so every solution of the DAE solves~\eqref{eq:local_reduced_ode}. Hence every consistent initial condition generates a unique local solution of the DAE, and the solution satisfies the junction balance for all $t$. One differentiation of the junction balance was needed to determine the junction heads. Since $\bd^J$ is twice continuously differentiable, the heads are continuously differentiable along the solution, so a second differentiation gives an ordinary differential equation for them. The differentiation index is therefore two~\cite[Def.~2.2.2 and Sec.~2.5.2]{Brenan1996}. In the variables $\bx_1=\bz$ and $\bx_2=\bp^J$, with $\bp^R=\bar{\bp}^R$ substituted, the system has the form $\dot{\bx}_1=\bfm(\bx_1,\bx_2,t)$, $\mathbf{0}=\bg(\bx_1,t)$ with $\bg_{\bx_1}\bfm_{\bx_2}=\bL_J$ nonsingular. This is the Hessenberg form of index two~\cite[Eq.~(4.36)]{kunkel2006dae}, for which the same existence and uniqueness statement is available in~\cite[Sec.~2.5]{Brenan1996} and~\cite{Stechlinski2021}. The solution can be continued on $\mathcal{I}$ as long as Assumption~\ref{ass:regular_fixed_mode} holds, since the argument can be repeated from any point of the solution.
\end{proof}

\section{Nondifferentiability at Status Changes}\label{app:proof_switch_nondiff}
The following result shows that an instantaneous pump or valve switch makes the hydraulic trajectory nondifferentiable at the switching instant, which motivates the smoothed DAE of Section~\ref{sec:smoothed_switching}.
\begin{theorem}[Instantaneous status changes destroy differentiability]\label{thm:switch_nondiff}
Let $t_s$ be a time at which a pump or valve status or setpoint changes. Assume Assumption~\ref{ass:regular_fixed_mode} holds on the time intervals immediately before and after $t_s$, and let $\mathbf{F}^-$ and $\mathbf{F}^+$ be the reduced vector fields obtained from~\eqref{eq:local_reduced_ode} on those two intervals. Assume no impulsive reset is applied, so the differential state has the common value $\bz_s=\bz(t_s^-)=\bz(t_s^+)$. If
\begin{equation}\label{eq:switch_kink_condition}
\mathbf{F}^-(\bz_s,t_s)\neq \mathbf{F}^+(\bz_s,t_s),
\end{equation}
then $\bz$ is not differentiable at $t_s$. In addition, the algebraic state may jump if the algebraic solution maps on the two sides of $t_s$ do not agree at $\bz_s$.
\end{theorem}
\begin{proof}
Apply Theorem~\ref{thm:local_index_one} separately on $[t_0,t_s)$ and $(t_s,t_1]$, where the residual maps and active constraints differ on the two intervals. On each side, the differential state obeys the reduced ODE~\eqref{eq:local_reduced_ode}, namely $\dot{\bz}=\mathbf{F}^-(\bz,t)$ for $t<t_s$ and $\dot{\bz}=\mathbf{F}^+(\bz,t)$ for $t>t_s$. Both vector fields are continuous at the common state value $\bz_s$ by Assumption~\ref{ass:regular_fixed_mode} on each side. Hence the one-sided derivatives of $\bz$ at $t_s$ exist and equal
\[
\dot{\bz}(t_s^-)=\mathbf{F}^-(\bz_s,t_s),\qquad
\dot{\bz}(t_s^+)=\mathbf{F}^+(\bz_s,t_s).
\]
If the two-sided derivative existed at $t_s$, the one-sided values would coincide. Condition~\eqref{eq:switch_kink_condition} forces them to differ, so $\bz$ fails to be differentiable at $t_s$. The algebraic variables are recovered through the implicit maps on each side $\boldsymbol{\pi}_y^\pm$ supplied by Theorem~\ref{thm:local_index_one}. Thus
\[
\by(t_s^\pm)=\boldsymbol{\pi}_y^\pm\!\big(\bz_s,\bu(t_s^\pm),\bd(t_s)\big),
\]
and $\by$ is continuous at $t_s$ if and only if these two values agree. Otherwise $\by$ exhibits a finite jump while $\bz$ stays continuous.
\end{proof}

\section{Proof of Theorem~\ref{thm:smoothed_switching}}\label{app:proof_smoothed_switching}
\begin{proof}
\textit{(Local well-posedness.)} The quintic blend in~\eqref{eq:quintic_smoothstep} satisfies $\chi(0)=0$, $\chi(1)=1$, $\chi'(0)=\chi'(1)=0$, $\chi''(0)=\chi''(1)=0$, so $\bu_{\tau_s}\in C^2(\mathbb{R};\mathbb{R}^{n_u})$ and coincides with the scheduled values outside the smoothing windows. The path $t\mapsto\bu_{\tau_s}(t)$ stays in the neighborhood where Assumption~\ref{ass:regular_fixed_mode} holds, so Assumption~\ref{ass:regular_fixed_mode} holds along the path. The matrix $\bL_J$ does not depend on $\bu$, so the connectivity condition of Assumption~\ref{ass:regular_fixed_mode}\textit{(iii)} is unchanged by the smoothing, and the map $\boldsymbol{\pi}_y$ of Theorem~\ref{thm:local_index_one} depends on $\bu_{\tau_s}$ only through the link equations. Theorem~\ref{thm:local_index_one} then provides a continuously differentiable algebraic map $\boldsymbol{\pi}_y$ on a neighborhood of the trajectory, and the reduced vector field $\mathbf{F}_{\tau_s}(\bz,t)=\mathbf{F}(\bz,\bu_{\tau_s}(t),\bd(t))$ is jointly continuous in $(t,\bz)$ and locally Lipschitz in $\bz$ uniformly on compact time intervals. Existence and uniqueness of a local solution then follow from the same ODE argument as in the proof of Theorem~\ref{thm:local_index_one}, and the index two reduction of Theorem~\ref{thm:local_index_one} is identical.

\textit{(Convergence.)} Fix $\mathcal{T}_{\mathrm{c}}=[a,b]\subset(t_k,t_{k+1})$ with $a>t_k+\tau_s$, so $\bu_{\tau_s}(t)=\bar{\bu}_{k+1}=\bu(t)$ for all $t\in\mathcal{T}_{\mathrm{c}}$ and all $\tau_s\le a-t_k$. The reduced vector fields of the smoothed and instantaneous switch models therefore coincide on $\mathcal{T}_{\mathrm{c}}$
\[
\mathbf{F}_{\tau_s}(\bz,t)=\mathbf{F}(\bz,t),\quad t\in\mathcal{T}_{\mathrm{c}}.
\]
Choose $\bar\tau>0$ so that, for every $\tau_s\in(0,\bar\tau]$, both trajectories remain in the same closed ball $\mathsf{B}$ over $[t_k,b]$. Because $\mathbf{F}$ is continuously differentiable, compactness of $\mathsf{B}\times[t_k,b]$ yields constants $L,M>0$ independent of $\tau_s$ such that $\mathbf{F}(\cdot,t)$ is $L$-Lipschitz on $\mathsf{B}$ and $\|\partial_{\bu}\mathbf{F}\|\le M$ there. Define $e(t)=\bz_{\tau_s}(t)-\bz(t)$. For $t\in[t_k,t_k+\tau_s]$ the two vector fields differ only through the controls, so $\|\dot e\|\le L\|e\|+M\|\bu_{\tau_s}(t)-\bar{\bu}_{k+1}\|$, while on $[t_k+\tau_s,b]$ one has $\|\dot e\|\le L\|e\|$. Since $\|\bu_{\tau_s}(t)-\bar{\bu}_{k+1}\|\le\|\bar{\bu}_{k+1}-\bar{\bu}_k\|$ on the smoothing window, $\|\dot e\|\le L\|e\|+M\|\bar{\bu}_{k+1}-\bar{\bu}_k\|$ there. Applying the integral comparison lemma~\cite[Lem.~2.7]{teschl2012ode}, which states that any nonnegative scalar function satisfying $f(t)\le \int_{t_k}^{t}(Lf(s)+c)\,ds$ obeys $f(t)\le(c/L)(e^{L(t-t_k)}-1)$, with $c=M\|\bar{\bu}_{k+1}-\bar{\bu}_k\|$ and starting from $e(t_k)=\mathbf{0}$ gives
\[
\|e(t_k+\tau_s)\|\le M\|\bar{\bu}_{k+1}-\bar{\bu}_k\|\,\tau_s\, e^{L\tau_s},
\]
where the bound $(e^{L\tau_s}-1)/L\le \tau_s e^{L\tau_s}$ has been used. Propagating to $t\in\mathcal{T}_{\mathrm{c}}$ with the same estimate yields $\|e(t)\|\le M\|\bar{\bu}_{k+1}-\bar{\bu}_k\|\,\tau_s\,e^{L(b-t_k)}\to 0$ uniformly on $\mathcal{T}_{\mathrm{c}}$ as $\tau_s\to 0$. Continuity of the algebraic map $\boldsymbol{\pi}_y$ then gives $\sup_{t\in\mathcal{T}_{\mathrm{c}}}\|\by_{\tau_s}(t)-\by(t)\|\to 0$, hence uniform convergence of the full state.
\end{proof}

\section{Proof of Theorem~\ref{thm:graph_linear_dae}}\label{app:proof_graph_linear_dae}
\begin{proof}
Linearizing~\eqref{eq:graph_form_dae_links} around the equilibrium gives
\begin{align*}
\boldsymbol{\Lambda}_q\Delta\dot{\bq}
= {}&
\bN_J^\top\Delta\bp^J+\bN_A^\top\Delta\bp^A+\bN_R^\top\Delta\bp^R\\
&-\bK\Delta\bq-\partial_{\bu}\boldsymbol{\eta}(\bar{\bq},\bar{\bu})\Delta\bu.
\end{align*}
The reservoir equation gives $\Delta\bp^R=\mathbf{0}$. Linearizing the junction and tank equations gives $\mathbf{0}=\bN_J\Delta\bq-\Delta\widetilde{\bd}^{J}$ and $\bA^A\Delta\dot{\bp}^{A}=-\bN_A\Delta\bq+\Delta\widetilde{\bd}^{A}$. Collecting these three equations yields~\eqref{eq:linearized_graph_dae}-\eqref{eq:linear_dae_blocks}. The matrices $\bB_u$ and $\bB_d$ are the corresponding input blocks.
a

At equilibrium, the link equation reduces to $\bN^\top\bar{\bp}=\boldsymbol{\eta}(\bar{\bq},\bar{\bu})$. Linearizing this static relation gives $\bN^\top\Delta\bp=\bK\Delta\bq$, hence $\Delta\bq=\bK^{-1}\bN^\top\Delta\bp=\bW\bN^\top\Delta\bp$. Multiplication by $\bN$ gives $\Delta\boldsymbol{\iota}=\bN\Delta\bq=\bN\bW\bN^\top\Delta\bp$, which proves~\eqref{eq:weighted_laplacian}.

For the energy identity, observe that the Lyapunov function of Theorem~\ref{thm:graph_linear_dae} is exactly
\[
V=\frac{1}{2}\Delta\bq^\top\boldsymbol{\Lambda}_q\Delta\bq
+\frac{1}{2}\Delta\bp^{A\top}\bA^A\Delta\bp^A,
\]
because the junction head block of $\bE_h$ is zero. Along an unforced consistent trajectory, $\bN_J\Delta\bq=\mathbf{0}$ and $\bA^A\Delta\dot{\bp}^{A}=-\bN_A\Delta\bq$. Therefore
\begin{align*}
\dot V
={}&
\Delta\bq^\top
\left(\bN_J^\top\Delta\bp^J+\bN_A^\top\Delta\bp^A-\bK\Delta\bq\right)\\
&{}+\Delta\bp^{A\top}(-\bN_A\Delta\bq)\\
={}&
(\bN_J\Delta\bq)^\top\Delta\bp^J
-\Delta\bq^\top\bK\Delta\bq=
-\Delta\bq^\top\bK\Delta\bq,
\end{align*}
which proves the energy identity in~\eqref{eq:linear_dae_energy}. It remains to identify the largest invariant set inside $\{\dot V=0\}$. Since $\bK\succ 0$, $\dot V=0$ forces $\Delta\bq=\mathbf{0}$, and the tank equation then gives $\Delta\dot{\bp}^A=\mathbf{0}$, so $\Delta\bp^A$ is constant. With $\Delta\bq=\mathbf{0}$ and $\Delta\bp^R=\mathbf{0}$, the linearized link relation $\bN^\top\Delta\bp=\bK\Delta\bq$ reduces to $\bN_J^\top\Delta\bp^J+\bN_A^\top\Delta\bp^A=\mathbf{0}$, i.e., $\bN_{\mathrm{nr}}^\top\Delta\bp_{\mathrm{nr}}=\mathbf{0}$, where $\bN_{\mathrm{nr}}=[\bN_J;\bN_A]$ is the incidence matrix restricted to nonreservoir nodes and $\Delta\bp_{\mathrm{nr}}=\{\Delta\bp^J,\Delta\bp^A\}$. The rank of the incidence matrix of a directed graph equals the number of vertices minus the number of connected components~\cite[Thm.~8.3.1]{godsil2001algebraic}. Therefore the full incidence $\bN$ has rank $(n_J+n_A+n_R)-c$, where $c$ is the number of components of the active link graph. The left null space of $\bN$ is spanned by the indicator vectors of those components. By assumption, every component contains at least one reservoir, so removing one reservoir row per component eliminates one direction of the left null space per component, leaving a matrix of full row rank. The nonreservoir submatrix $\bN_{\mathrm{nr}}$ inherits this property, so $\bN_{\mathrm{nr}}^\top$ has trivial null space, forcing $\Delta\bp^J=\mathbf{0}$ and $\Delta\bp^A=\mathbf{0}$. The largest invariant set inside $\{\dot V=0\}$ is therefore $\{\mathbf{0}\}$. By LaSalle's invariance principle~\cite[Thm.~4.4]{khalil2002nonlinear}, every bounded trajectory of the unforced system is contained in the largest invariant subset of $\{\dot V=0\}$. Applied to the linearized system on the set of states satisfying the algebraic constraints, this yields local asymptotic stability of the equilibrium.
\end{proof}

\section{Proof of Theorem~\ref{thm:linearization_bounds}}\label{app:proof_linearization_bounds}
\begin{proof}
Because $\bJ_x$ is nonsingular and $\mathbf{r}_{\vartheta}$ is three times continuously differentiable, the implicit function theorem yields a unique three times continuously differentiable equilibrium map $\bx^\star(\boldsymbol{\vartheta})$ on a neighborhood of $\bar{\boldsymbol{\vartheta}}$. Differentiate the identity $\mathbf{r}_{\vartheta}(\bx^\star(\boldsymbol{\vartheta}),\boldsymbol{\vartheta})=\mathbf{0}$ with respect to $\boldsymbol{\vartheta}$ and evaluate at $\bar{\boldsymbol{\vartheta}}$ to obtain
$\bJ_x\bS_\vartheta+\bJ_\vartheta
=
\mathbf{0},$
which is equivalent to~\eqref{eq:equilibrium_sensitivity}.
Since the equilibrium map is at least twice continuously differentiable on a possibly smaller neighborhood, Taylor's theorem with a Lagrange remainder gives
\begin{align*}
\bx^\star(\bar{\boldsymbol{\vartheta}}+\Delta\boldsymbol{\vartheta})
= {}& \bar{\bx}
+\bS_\vartheta\Delta\boldsymbol{\vartheta}
+\mathbf{r}_x(\Delta\boldsymbol{\vartheta}),\\ & \|\mathbf{r}_x(\Delta\boldsymbol{\vartheta})\|
\le  \tfrac{L_x}{2}\|\Delta\boldsymbol{\vartheta}\|^2,
\end{align*}
which proves~\eqref{eq:equilibrium_shift_bound}. The local DAE matrix $\bA_h(\boldsymbol{\vartheta})$ is obtained by differentiating the fixed control DAE with respect to the state and then evaluating the result at the equilibrium $\bx^\star(\boldsymbol{\vartheta})$.
The state Jacobian $\partial_{\bx}\mathbf{r}_{\vartheta}$ is twice continuously differentiable because $\mathbf{r}_{\vartheta}$ is three times continuously differentiable, and its composition with the twice continuously differentiable equilibrium map is twice continuously differentiable in $\boldsymbol{\vartheta}$. Another Taylor expansion with a Lagrange remainder therefore gives
\begin{align*}
\bA_h(\bar{\boldsymbol{\vartheta}}+\Delta\boldsymbol{\vartheta})
= {}& \bA_h(\bar{\boldsymbol{\vartheta}})
+\nabla_{\boldsymbol{\vartheta}}\bA_h(\bar{\boldsymbol{\vartheta}})
\Delta\boldsymbol{\vartheta}
+\mathbf{R}_A(\Delta\boldsymbol{\vartheta}),\\
\|\mathbf{R}_A(\Delta\boldsymbol{\vartheta})\|
\le {}& \tfrac{L_A}{2}\|\Delta\boldsymbol{\vartheta}\|^2,
\end{align*}
which is exactly~\eqref{eq:linearization_matrix_bound}. Finally, enforcing $\frac{L_A}{2}\|\Delta\boldsymbol{\vartheta}\|^2\le \varepsilon_{\mathrm{lin}}\|\bA_h(\bar{\boldsymbol{\vartheta}})\|$ yields~\eqref{eq:linearization_validity_radius}. This completes the proof.
\end{proof}

\section{Proof of Theorem~\ref{thm:property_robustness}}\label{app:proof_property_robustness}
\begin{proof}
\textit{(Stability.)} Let $\Delta\bA=\bA_{\mathrm{red}}(\bar{\boldsymbol{\vartheta}}+\Delta\boldsymbol{\vartheta})-\bA_{\mathrm{red}}(\bar{\boldsymbol{\vartheta}})$. Since $\bA_{\mathrm{red}}(\bar{\boldsymbol{\vartheta}})=\bV\boldsymbol{\Lambda}\bV^{-1}$ is diagonalizable, the eigenvalue perturbation bound stating that every eigenvalue of the perturbed matrix lies within $\kappa(\bV)\|\Delta\bA\|$ of some nominal eigenvalue~\cite[Thm.~6.3.2]{hornjohnson2013matrix} gives, for any eigenvalue $\lambda$ of $\bA_{\mathrm{red}}(\bar{\boldsymbol{\vartheta}}+\Delta\boldsymbol{\vartheta})$,
\[
\min_j|\lambda-\lambda_j(\bA_{\mathrm{red}}(\bar{\boldsymbol{\vartheta}}))|
\le
\kappa(\bV)\|\Delta\bA\|,
\]
where $\kappa(\bV)=\|\bV\|\,\|\bV^{-1}\|$. Combining with~\eqref{eq:state_space_perturbation_bound_A} gives
\[
\Re\lambda\le -\alpha_s+\kappa(\bV)c_A\|\Delta\boldsymbol{\vartheta}\|,
\]
so every perturbed eigenvalue lies in the open left half plane whenever $\kappa(\bV)c_A\|\Delta\boldsymbol{\vartheta}\|<\alpha_s$. This proves exponential stability and the first bound in~\eqref{eq:stability_margin_perturbation} up to the $O(\|\Delta\boldsymbol{\vartheta}\|^2)$ remainder, which collects the higher-order terms from the local Taylor expansion of $\bA_{\mathrm{red}}(\boldsymbol{\vartheta})$ around $\bar{\boldsymbol{\vartheta}}$.

\textit{(Controllability.)} Let $\Delta\boldsymbol{\mathcal{K}}_N=\boldsymbol{\mathcal{K}}_N(\bar{\boldsymbol{\vartheta}}+\Delta\boldsymbol{\vartheta})-\boldsymbol{\mathcal{K}}_N(\bar{\boldsymbol{\vartheta}})$. The singular value perturbation inequality stating that the smallest singular value of a matrix changes by at most the spectral norm of the perturbation~\cite[Cor.~7.3.5]{hornjohnson2013matrix} gives
\[
\big|\sigma_{\min}(\boldsymbol{\mathcal{K}}_N(\bar{\boldsymbol{\vartheta}}+\Delta\boldsymbol{\vartheta}))-\sigma_{\min}(\boldsymbol{\mathcal{K}}_N(\bar{\boldsymbol{\vartheta}}))\big|
\le
\|\Delta\boldsymbol{\mathcal{K}}_N\|.
\]
Combining with~\eqref{eq:kalman_perturbation_bound} yields
\[
\sigma_{\min}\!\big(\boldsymbol{\mathcal{K}}_N(\bar{\boldsymbol{\vartheta}}+\Delta\boldsymbol{\vartheta})\big)
\ge
\sigma_c-c_{\mathcal{K}}(c_A+c_B)\|\Delta\boldsymbol{\vartheta}\|,
\]
which is~\eqref{eq:controllability_margin_perturbation}. The right-hand side is positive whenever $c_{\mathcal{K}}(c_A+c_B)\|\Delta\boldsymbol{\vartheta}\|<\sigma_c$. In that regime the controllability matrix retains full row rank and the perturbed reduced model is controllable over horizon $N$. This ends the proof. 
\end{proof}

\balance
\section{Numerical Settings of the Case Studies}\label{app:numerical_settings}
All studies use the nominal pump speeds, and the remaining settings are as follows. \textit{(i)} The fixed control validation in Section~\ref{sec:case_theorem1} is for one hour with $61$  samples and a smooth demand profile. \textit{(ii)} The smoothed control validation in Section~\ref{sec:case_theorem3} covers $24$~h with the EPANET step $\Delta t_E=300$~s, two DAE substeps per EPANET step, and a smoothing width $\tau_s=300$~s. The pump trip closes pump~\texttt{10} of \texttt{Net3} $70$~min into the run, and the DAE step is $60$~s before the trip, $2$~s for $15$~min after it, and $10$~s afterwards. \textit{(iii)} The energy decay study in Section~\ref{sec:case_theorem5} starts from a perturbation of size $\|\Delta\bx_h(0)\|_2=10^{-3}$ and has $2000$ implicit Euler steps. \textit{(iv)} The linearization study in Section~\ref{sec:case_theorem6} considers global roughness factors $\delta\in\{0.25,0.5,1,2,5,10\}\%$ and a finite difference step $h_\theta=10^{-4}$, and the margin study changes the roughness of each pipe separately by $\Delta\theta_i=5\%$. $\sigma_c$ is evaluated over $H=2$~h with a $60$~s step. \textit{(v)} The demand variation in Section~\ref{sec:case_operating_point_margins} is for a total \texttt{Net3} demand $D\in[200,800]$~L/s sampled every $0.5$~h, with $H=2$~h. \textit{(vi)} The pump allocation in Section~\ref{sec:case_actuators} is for a $60$~s step, $H=2$~h simulations, $\epsilon=10^{-6}$, and for six added pumps.

\end{document}